\documentclass[12pt]{article}
\usepackage{appendix}
\usepackage{amsmath, amsthm, amssymb}
\usepackage[mathscr]{euscript}
\usepackage{verbatim}
\numberwithin{equation}{section}
\usepackage{hyphenat}

\usepackage{enumerate}
\usepackage{setspace}
\usepackage{graphicx}
\usepackage{natbib}
\usepackage[bf, footnotesize]{caption2}

\newcommand{\derv}[1]{\frac{\partial}{\partial #1}}

\newcommand{\deriv}[2]{\frac{\partial #1}{\partial #2}}

\newcommand{\beqn}{\begin{equation}}
\newcommand{\eeqn}{\end{equation}}
\newcommand{\beqnar}{\begin{eqnarray}}
\newcommand{\eeqnar}{\end{eqnarray}}

\newtheorem{theorem}{Theorem}[section]

\newtheorem{proposition}[theorem]{Proposition}

\newenvironment{remark}[1][Remark]{\begin{trivlist}
\item[\hskip \labelsep {\bfseries #1}]}{\end{trivlist}}

\title{{Rossby Wave Green's Functions in an Azimuthal Wind}\\ 
\author{G.M. Webb${}^1$ 
  C. T. Duba${}^{2}$ 
  and 
Q. Hu${}^{1,3}$\\
\mbox{}\\
${}^1$CSPAR, The University of Alabama in Huntsville,\\ Huntsville,
AL 35805, USA\\
\mbox{}\\
${}^2$Department of Mathematics and Statistics,\\
Durban University of Technology,\\ Steve Biko Campus, Durban 4001 South Africa\\
\mbox{}\\
${}^3$Department of Space Science, The University of Alabama in Huntsville,
\\ Huntsville,
AL 35899, USA}}
\date{Submitted to Geophysical and Astrophysical Fluid Dynamics, \today}
\begin{document}
\maketitle

\begin{abstract}
Green's functions for Rossby waves in an azimuthal wind are obtained, in which 
the stream-function $\psi$ depends on $r$, $\phi$ and $t$, where $r$ 
is cylindrical radius and $\phi$ is the azimuthal angle in the $\beta$-plane 
relative to the easterly direction, in which the $x$-axis points east 
and the $y$-axis points north. 
The Rossby wave Green's function with no wind is obtained 
using Fourier transform methods, and is related to the previously known 
Green's function obtained for this case, which has a different 
but equivalent form to the Green's function obtained in the present paper. 
We emphasize the role of the wave eikonal solution, which plays an 
important role in the form of the solution.  
The corresponding Green's function for a rotating wind with azimuthal 
wind velocity ${\bf u}=\Omega r{\bf e}_\phi$ ($\Omega=$const.) is also obtained by Fourier methods, 
in which the advective rotation operator in position space is transformed to a rotation 
operator in ${\bf k}$ transform space. The finite Rossby deformation radius is included in the analysis.
The physical characteristics of the Green's functions are delineated
and applications are discussed. 
In the limit as $\Omega\to 0$, the rotating wind Green's function reduces to the 
Rossby wave Green function with no wind. 
\end{abstract}

\section{Introduction}
Rotating planetary atmospheres admit the propagation of a range of different wave modes.
At frequencies $\omega$ above the acoustic-gravity wave cut-off 
$\Omega_s=c_s/(2h)$ ($c_s$ is the adiabatic sound speed and 
$h$ is the gravitational scale height),  the equations admit high frequency acoustic-gravity 
waves. Below the acoustic cut-off, but above the Brunt-V{\"a}is{\"a}l{\"a}  
frequency $N$, the waves are evanescent. Below the 
Brunt frequency $N$, and above the Coriolis frequency $f_0$, there exist dispersive, anisotropic, 
inertial-gravity waves.  Rossby waves propagate at frequencies ($\omega<f_0$) below the 
Coriolis frequency $f_0$ 
(e.g. Eckart (1960), Pedlosky (1987)). 

Rossby waves on a rotating planet are described in classical texts by Gill (1982), 
Pedlosky (1987) and Vallis (2006). Rossby waves are planetary scale waves, that arise 
from the latitudinal variation of the vertical component of the Coriolis force, 
known as the $\beta$-effect, which is closely related to the concept of 
the conservation of potential vorticity (e.g. Vallis (2006), p. 178-183). The dispersion 
and anisotropic propagation of these waves are best understood by using the wave normal curve in wave number 
space (${\bf k}=(k_x,k_y)^T$-space) at a fixed frequency $\omega$ (Longuet-Higgins (1964), Lighthill (1978))
 which consists of a circle in ${\bf k}$-space,
with center displaced westward ($k_x<0$) along the $k_x$-axis, 
with center $(k_x,k_y)=(-\beta/(2\omega),0)$, and 
with diameter $\beta/\omega$. The phase velocity of the wave is westward, but the group velocity 
can be both eastward or westward, depending on the wave number 
(e.g. Duba and McKenzie (2012), McKenzie (2014)). Duba et al. (2014) describe Rossby wave patterns in 
zonal and meridional winds. 

McKenzie and Webb (2015) investigated Rossby waves in a rotating wind. They obtained analytical 
solutions for the waves, in which  the stream function $\psi$ depends on $(r,\phi,t)$ 
where $r$ is cylindrical radius, $\phi$ is the azimuthal angle measured in the $\beta$-plane 
relative to the easterly direction (in the horizontal $\beta$-plane, the $x$-axis points east, and 
the $y$-axis points north). They investigated solutions for $\psi(r,\phi,t)$ that are $2\pi$-periodic, in $\phi$,
and consist of a superposition of Bessel functions of the form:
\begin{equation}
\psi(r,\phi,t)=\sum_{n=-\infty}^\infty a_n J_n(kr)\exp[i(\omega t-n\phi)]. \label{eq:1.1}
\end{equation}
It was shown that solutions exist in which the expansion coefficients $\{a_n\}$ ($n$ an integer), 
satisfy three-term recurrence relations. The recurrence relations were analyzed by using 
Fourier-Floquet analysis, which shows that the solutions for $\psi$ can be interpreted in terms of 
three-wave interactions. The recurrence relations have solutions in terms of Bessel functions 
in which the $a_n\propto J_{n-{\bar\omega}}(a)$ where ${\bar\omega}=\omega/\Omega$ and $a=\beta/(k\Omega)$
where $\Omega$ is the angular velocity of the wind about the local vertical $z$-axis. The solutions for the $\{a_n\}$
when substituted in (\ref{eq:1.1}) give rise to Neumann series that can be summed up to give the solution 
for  $\psi$ in terms of a single Bessel function. The solutions were used to illustrate the properties
of Rossby waves in a rotating wind. In the long wavelength limit these solutions reduce to the classical 
Rossby wave solution with westward phase velocity, $V_p=-\beta k_x/k^3$, where $k^2=k_x^2+k_y^2$ 
and with dispersion equation $\omega=-\beta k_x/k^2$.  

In this paper, we revisit the problem of  Rossby waves in a rotating wind 
investigated by McKenzie and Webb (2015), 
by using a direct Fourier transform approach. 
We obtain the Rossby wave Green's function for a Dirac delta distribution 
source $Q=N\delta(x)\delta(y)\delta(t)$ in the stream function equation for Rossby waves, 
where $N$ is a normalization constant.   
 We study both the case of Rossby waves in the absence of a rotating wind (i.e. $\Omega=0$) 
 and also the case of a rotating wind. Veronis (1958) obtained the Green's function for the case 
of no azimuthal wind ($\Omega=0$) 
and also  included the effect of a finite Rossby deformation radius (the Rossby 
deformation radius is given by $R_d=(gH)^{1/2}/f_0$ where $H$ is the depth of the fluid, $g$ 
is the acceleration due to gravity and $f_0$ is the Coriolis frequency).  
 Veronis (1958) only investigated in detail, the case $R_d>>L$ where $L$ is the scale length of 
waves of interest, for which the effects of $R_d$ can be neglected.  We obtain a different,  
but equivalent form of the Green's function obtained by Veronis. Our form of the Green's function applies 
both in the case $k_d=1/R_d\neq 0$ and for the case $k_d=0$. Veronis (1958) only investigated in detail
the case $k_d=0$.  The wave eikonal (i.e. the wave dispersion 
equation, written as a first order partial differential equation for the wave phase $S$) 
plays an important role in the analysis. This is to be expected from the analysis of the group velocity 
of the waves based on the wave normal diagram and the method of stationary phase, which is used 
for example by McKenzie (2014) in the analysis of Rossby waves for a finite Rossby deformation radius. 
  Rhines (2003) discusses the Green's function 
for a harmonic time source
located at $(x,y)=(0,0)$, and refers to a 
more complete analysis of Rossby wave Green's functions by Dickinson (1968, 1969a,b).

The basic Rossby wave model is described in Section 2. 
In Section 3, we derive the Rossby wave Green's function for the 
case of no wind ($\Omega=0$) and with source term $Q\propto\delta(x)\delta(y)\delta(t)$. 
In Section 4, this solution is 
generalized to the case of a Rossby wave in a rotating wind, with azimuthal wind velocity 
${\bf U}=\Omega r {\bf e}_\phi$,  
 where ${\bf e}_\phi$ is a unit vector in the azimuthal direction. 
 The solutions in Section 3 and Section 4 are obtained by Fourier transform methods. 
Section 5 presents illustrative examples and applications of the solutions are discussed. 
Appendix A discusses solutions of the wave eikonal equation for Rossby waves in the non-rotating 
wind case. The wave eikonal equation is the nonlinear first order partial differential equation 
that results from setting $\omega=-S_t$ and ${\bf k}=\nabla S$ in the Rossby wave dispersion 
equation. Nonlinear solutions of the wave eikonal equation are identified in the  
Green's function solution obtained by Fourier analysis. Appendix B and Appendix C describe the 
relationship between the Veronis (1958)  Green's function and the present analysis. 
Section 6 concludes with a summary and discussion.

\section{The Model}
We use the classical linearized Rossby wave equation for the stream function $\psi(x,y,t)$ 
on a $\beta$-plane, which has the form:
\begin{equation}
\left(\derv{t}+{\bf U}{\bf\cdot}\nabla\right)(\nabla_\perp^2\psi-k_d^2\psi)
+\beta\deriv{\psi}{x}=Q, \label{eq:2.1}
\end{equation}
where
\begin{equation}
\zeta=(\nabla\times{\bf u}){\bf\cdot}{\bf e}_z=\nabla_\perp^2\psi=\psi_{xx}+\psi_{yy}, \label{eq:2.2}
\end{equation}
is the fluid vorticity in the local vertical direction in the $\beta$-plane. 
 In the $\beta$-plane approximation, the  $z$-axis is the local vertical direction perpendicular to the 
Earth's surface. The $\beta$-parameter is related to the Coriolis parameter $f$ via the equations:
\begin{equation}
f=2\Omega_E\sin\theta_0+\beta y,\quad  \beta=2\frac{\Omega_E}{R}\cos\theta_0, \label{eq:2.3}
\end{equation}
where the $\beta$-plane is centered on latitude $\theta_0$ on a planet of radius $R$, rotating with an angular 
frequency $\Omega_E$ and the $x$ and $y$ coordinates are local Cartesian coordinates pointing east 
and north respectively. The $\beta$-effect is maximal at the equator ($\theta=0$) 
and zero at the poles ($\theta=\pm\pi/2$). The local fluid velocity perturbation 
\begin{equation}
{\bf u}=(u,v,0)={\bf e}_z\times\nabla\psi=(-\psi_y,\psi_x,0). \label{eq:2.3a}
\end{equation}
The wavenumber $k_d$ in (\ref{eq:2.1}) is the inverse of the Rossby deformation radius $R_d$, i.e. 
\begin{equation}
k_d=\frac{1}{R_d},\quad R_d=\sqrt{gH/f_0^2}=\frac{c}{f_0}\quad\hbox{where}\quad c=\sqrt{gH}, \label{eq:2.3b}
\end{equation}
is the shallow water speed, $g$ is the acceleration due to gravity,
 $f_0=2\Omega_E\sin\theta_0$ is the Coriolis parameter, 
 and $H$ is the effective fluid depth in a shallow water approximation (see e.g. Veronis 1958; 
Vallis 2006; Pedlosky 1987). The stream function $\psi$ is related to the deviation of the height 
$\eta$ by the equation $\eta=-(f_0/g)\psi$.  In  (\ref{eq:2.1}) the geostrophic approximation is used, 
and the geostrophic response time is assumed to be much larger than a half-pendulum day, i.e. 
 $\eta_x>>(1/f_0) \partial^2\eta/\partial t\partial y$
where $\eta$ is the variable fluid depth (e.g. Veronis, 1958). McKenzie (2014) gives a more general version
of the Rossby wave equation (\ref{eq:2.1}) in terms of $\eta$ which includes the topographic contribution  
to the $\beta$ effect. This amounts to replacing $\beta\psi_x$ by 
$(\boldsymbol{\beta}\times\nabla\psi){\bf\cdot}{\bf e}_z$ in (\ref{eq:2.1}) where 
\begin{equation}
\boldsymbol{\beta}=\boldsymbol{\beta}_t+\boldsymbol{\beta}_c,\quad 
 \quad \boldsymbol{\beta}_t=-f_0\left(\frac{1}{c^2}\deriv{c^2}{x},\frac{1}{c^2}\deriv{c^2}{y}\right),
\quad \boldsymbol{\beta}_c=\left(0,\deriv{f}{y}\right), 
\label{eq:2.3c}
\end{equation}
in which $c_o=\sqrt{gH}$ is the shallow water speed at the $o$ reference level.   

We assume that the background wind velocity ${\bf U}$ in the $\beta$-plane is azimuthal and has the form:
\begin{equation}
{\bf U}=U_\phi {\bf e}_\phi=r\Omega(r) {\bf e}_\phi, \label{eq:2.4}
\end{equation}
where ${\bf e}_\phi=(-\sin\phi,\cos\phi,0)^T$ is the unit vector in the azimuthal direction. We also use 
the Cartesian coordinates $x$ and $y$, where
\begin{equation}
x=r\cos\phi, \quad y=r\sin\phi. \label{eq:2.5}
\end{equation}

The main aim of the present paper is to determine the Green's function of (\ref{eq:2.1}) 
both for the case $\Omega\neq 0$  and also for the case $\Omega=0$ (no-wind) case. The no wind 
case ($\Omega=0$) was   
investigated by Veronis (1958). The Green's function for $\Omega=0$ which we obtain has a different 
form than that obtained by Veronis (1958). 
The Green's function solutions of (\ref{eq:2.1}) have source term  $Q=N\delta(x)\delta(y)\delta(t)$ on the 
right-hand side. For the case of a constant angular velocity $\Omega=const.$, the Rossby wave equation (\ref{eq:2.1})
can be written in the form:
\begin{equation}
\left[\derv{t}+\Omega\left(-y\derv{x}+x\derv{y}\right)\right]\left(\psi_{xx}+\psi_{yy}-k_d^2\psi\right) 
+\beta \psi_x=Q, \label{eq:2.6}
\end{equation}
where
\begin{equation}
Q=N \delta(x)\delta(y)\delta (t), \label{eq:2.7}
\end{equation}
is the source term for the Green's function solution where $N$ is a normalization constant. 
 Veronis (1958) obtained the Green's function 
with a delta function source term for the case $\Omega=0$ for the deviation of the height 
$\eta$ of the fluid layer. 
An alternative source term:
\begin{equation}
Q=Q_H=N \delta(x)\delta(y)\exp(-i\omega t), \label{eq:2.8}
\end{equation}
is sometimes used (e.g. Rhines (2003)), in order to determine the characteristics of 
Rossby waves with driving frequency $\omega$. Our main emphasis is on Green's function solutions 
with $\delta$-function source term (\ref{eq:2.7}). 

An alternative form of the Rossby wave equation (\ref{eq:2.1}) or (\ref{eq:2.6}) is to use cylindrical 
polar coordinates $(r,\phi)$ to describe the $\beta$-plane. In terms of $(r,\phi,t)$ the Rossby wave equation has 
the form:
\begin{equation}
\left(\derv{t}+\Omega\derv{\phi}\right)\left[\frac{1}{r}\derv{r}\left(r\deriv{\psi}{r}\right)
+\frac{1}{r^2}\frac{\partial^2\psi}{\partial\phi^2}-k_d^2\psi\right]
+\beta\left(\cos\phi\deriv{\psi}{r}-\frac{\sin\phi}{r}\deriv{\psi}{\phi}\right)=Q. \label{eq:2.9}
\end{equation}
We restrict our analysis to the $\Omega=const.$ case. 
The Rossby wave equation (\ref{eq:2.9}) with $k_d=0$ was used by McKenzie and Webb (2015) 
in their work on Rossby waves in a rotating wind.

\section{Green's function, with no wind ($\Omega=0$)}
\begin{proposition}\label{prop31}
The Green's function solution of the Rossby wave equation (\ref{eq:2.6}) for the case of no azimuthal wind, 
with $\Omega=0$ and with a Dirac delta function source (\ref{eq:2.7}) is given by 
the formula:
\begin{equation}
\psi_G=-\frac{N}{2\pi}\int_0^\infty \frac{k}{k^2+k_d^2} J_0(A)\ dk, \label{eq:nowind1}
\end{equation}
where $J_0(A)$ is a Bessel function of the first kind of order zero. The argument $A$ is given by:
\begin{equation}
A=\left[\left(k x+\frac{\beta k  t}{k^2+k_d^2}\right)^2+k^2 y^2\right]^{1/2}. 
\label{eq:nowind2}
\end{equation}
\end{proposition}

\begin{proof}
To solve the Rossby wave equation (\ref{eq:2.6}) with source term 
$Q=N \delta(x)\delta(y)\delta(t)$ and with $\Omega=0$ (no wind), we introduce the Fourier transform:
\begin{equation}
\bar{\psi}={\cal F}(\psi)=
\int_{-\infty}^\infty dt\int_{-\infty}^\infty dx\int_{-\infty}^\infty dy 
\exp\left[i({\bf k}{\bf\cdot}{\bf x}-\omega t)\right]\psi({\bf x},t). \label{eq:3.1}
\end{equation}
Taking the transform of (\ref{eq:2.6}) and noting that 
\begin{equation}
{\cal F}(\psi_t)=i\omega\bar{\psi},\quad {\cal F}(\psi_x)=-ik_x\bar{\psi}, 
\quad {\cal F}(\psi_y)=-ik_y\bar{\psi}, \label{eq:3.2}
\end{equation}
we obtain the Fourier transform equation:
\begin{equation}
-i\left[\omega \left(k^2+k_d^2\right)+k_x\beta\right]\bar{\psi}=N, \label{eq:3.3}
\end{equation}
with solution:
\begin{equation}
\bar{\psi}=\frac{iN }{[\omega (k^2+k_d^2)+k_x\beta]}\quad \hbox{where}\quad k^2=k_x^2+k_y^2. 
\label{eq:3.4}
\end{equation}

Fourier inversion of (\ref{eq:3.4}) gives 
\begin{equation}
\psi({\bf x},t)=\frac{N}{(2\pi)^3}
\int_{-\infty}^\infty d\omega\int_{-\infty}^\infty dk_x
\int_{-\infty}^\infty dk_y\frac{i \exp[i(\omega t-{\bf k}{\bf\cdot x})]}
{[\omega (k^2+k_d^2)+k_x\beta]}, \label{eq:3.5}
\end{equation}
as the required solution for $\psi({\bf x},t)$. 

Using cylindrical polar coordinates in $(x,y)$ and $(k_x,k_y)$ space, i.e.
\begin{equation}
(x,y)=r(\cos\phi,\sin\phi),\quad (k_x,k_y)=k(\cos\Phi,\sin\Phi), \label{eq:3.6}
\end{equation}
the solution (\ref{eq:3.5}) takes the form:
\begin{equation}
\psi({\bf x},t)=\frac{i N}{(2\pi)^3}\int_{-\infty}^\infty d\omega 
\int_0^\infty dk \frac{k}{k^2+k_d^2} \int_0^{2\pi} d\Phi 
\frac{\exp[i(\omega t-k r\cos(\Phi-\phi))]}{(\omega-\omega_p)}, \label{eq:3.7}
\end{equation}
where 
the integrand in (\ref{eq:3.7}) has a pole in the complex $\omega$-plane at
\begin{equation}
\omega=\omega_p=-\frac{\beta k\cos\Phi}{k^2+k_d^2}. \label{eq:3.8}
\end{equation}
Because $t>0$, is required by causality, the contour integral in (\ref{eq:3.7}) must be closed in 
the $\rm{Im}(\omega)>0$ half plane and the contour is deformed below the pole on the $\rm{Re}(\omega)$ 
axis (causality requires outgoing waves, by the Lighthill causality condition). Using the Residue theorem 
we obtain:
\begin{equation}
\psi=-\frac{N}{4\pi^2}\int_0^\infty dk\ \frac{k}{k^2+k_d^2} \int_0^{2\pi}d\Phi 
\exp\left(-i\left[\frac{\beta k \cos\Phi}{k^2+k_d^2} t+k r\cos(\Phi-\phi)\right]\right). 
\label{eq:3.9}
\end{equation}
The argument of the exponential function in (\ref{eq:3.9}) can be written in the form: 
\begin{align}
\zeta=&-i\left[\frac{\beta (k\cos\Phi) t}{k^2+k_d^2}
+kr\left[\cos\Phi\cos\phi+\sin\Phi\sin\phi\right]\right]\nonumber\\
=&-i\left[\left(kx+\frac{\beta kt}{k^2+k_d^2}\right)\cos\Phi+ky\sin\Phi\right]\nonumber\\
=&-i A\sin(\Phi+\epsilon) \label{eq:3.10}
\end{align}
where
\begin{equation}
A\sin\epsilon=k x+\frac{\beta k t}{k^2+k_d^2},\quad A\cos\epsilon=k y. \label{eq:3.11}
\end{equation}
Without loss of generality, we assume $A>0$. Solving (\ref{eq:3.11}) for $A$ and $\epsilon$ we obtain:
\begin{align}
A=&\left[\left(k x+\frac{\beta k t}{k^2+k_d^2}\right)^2 
+k^2 y^2\right]^{1/2}, \nonumber\\
\tan\epsilon=&\frac{k x+\beta k t/(k^2+k_d^2)}{k_\perp y}. \label{eq:3.12}
\end{align}

Using (\ref{eq:3.10})-(\ref{eq:3.11}) in (\ref{eq:3.9}) gives:
\begin{equation}
\psi=-\frac{N}{4\pi^2}\int_0^\infty dk \frac{k}{k^2+k_d^2} \int_0^{2\pi}d\Phi \exp[iA\sin(-\Phi-\epsilon)].
\label{eq:3.13}
\end{equation}

Using the Bessel function generating expansion:
\begin{equation}
\exp(iz\sin\theta)=\sum_{n=-\infty}^\infty \exp(in\theta) J_n(z), \label{eq:3.14}
\end{equation}
(Abramowitz and Stegun (1965): set $\tau=\exp(i\theta)$ in formula 9.1.41 p. 361). In the application of 
(\ref{eq:3.14}) to (\ref{eq:3.13}) we set $\theta=-\Phi-\epsilon$. Carrying out the integral 
over $\Phi$, only the $n=0$ term survives and we get:
\begin{equation}
\psi=\psi_G=-\frac{N}{2\pi} \int_0^\infty dk \frac{k}{k^2+k_d^2} J_0(A), \label{eq:3.15}
\end{equation}
as the Rossby wave Green's function $\psi_G$ 
where $A$ is given by (\ref{eq:3.12}). 
\end{proof}

The Green's function solution $\psi_G$ can be written in the form:
\begin{equation}
\psi_G=-\frac{N}{2\pi} \int_0^\infty dk \frac{k}{k^2+k_d^2} \psi^A\quad\hbox{where}\quad \psi^A=J_0(A). 
\label{eq:3.17}
\end{equation}
One can show that $\psi^A$  satisfies the 
Rossby wave equation (\ref{eq:2.1}) with $U=0$, i.e. 
\begin{equation}
\derv{t}\left(\psi^A_{xx}+\psi^A_{yy}-k_d^2\psi^A\right)+\beta \deriv{\psi^A}{x}=0. \label{eq:3.18}
\end{equation}
Assuming that it is valid to interchange the order of differentiation and integration 
in (\ref{eq:3.15}) it follows that $\Psi_G$ also satisfies the Rossby wave equation (\ref{eq:3.18}) 
(a special analysis is obviously needed near the source point). 

Consider the limit as the Rossby deformation radius $R_d\to\infty$ and $k_d\to 0$. 
In this limit the integrand in (\ref{eq:3.15}) is still integrable.  In the limit as $k\to 0$, 
$A\to\infty$ and
\begin{equation}
\frac{J_0(A)}{k}\sim \left(\frac{2}{\pi A_0}\right)^{1/2}
\cos\left(\frac{A_0}{k}-\frac{\pi}{4}\right)\frac{1}{\sqrt{k}}\quad 
\hbox{where}\quad A_0=\sqrt{\beta t}. \label{eq:3.19}
\end{equation}
Thus $|J_0(A)/k|\sim 1/\sqrt{k}$ which is integrable as $k\to 0$. 

\subsection{Wave Eikonal}
In this section we point out that the phase $A$ of the Bessel function $J_0(A)$ in the 
Green's function (\ref{eq:3.1}) for a fixed $k$ can be related to the dispersion equation 
for Rossby waves:
\begin{equation}
D(\omega,{\bf k})=\omega(k_x^2+k_y^2+k_d^2)+k_x\beta=0. \label{eq:3.21}
\end{equation}
In the JWKB approximation (e.g. Whitham (1974)), it is customary to introduce the wave phase or 
wave eikonal $S({\bf x},t)$ such that 
\begin{equation}
\omega=-S_t\quad \hbox{and}\quad {\bf k}=\nabla S. \label{eq:3.22}
\end{equation}
Using the identifications (\ref{eq:3.22}), the Rossby wave dispersion equation reduces 
to a first order partial differential equation for $S$:
\begin{equation}
S_t\left(S_x^2+S_y^2+k_d^2\right)-\beta S_x=0, \label{eq:3.23}
\end{equation}
which is known as the wave eikonal equation. An obvious plane wave solution of (\ref{eq:3.23})
is the plane wave solution:
\begin{equation}
S=(k_x x+k_y y)-\omega t, \label{eq:3.24}
\end{equation}
(Note that we could equally well define the wave phase as $S=\omega t-{\bf k}{\bf\cdot}{\bf x}$
by changing $\omega\to -\omega$ and ${\bf k}\to -{\bf k}$). 
Substitution of the solution ansatz (\ref{eq:3.24}) into the eikonal equation (\ref{eq:3.23})
gives the dispersion equation (\ref{eq:3.21}) for Rossby waves. The solution (\ref{eq:3.24})
is a linear solution of the wave eikonal equation. However, the nonlinear eikonal equation 
(\ref{eq:3.23}) also has nonlinear solutions which can be thought of as envelope solutions 
of the family of plane wave solutions (\ref{eq:3.24}) (e.g. Sneddon (1957), Courant and Hilbert (1989)). 
Below 
we show that the phase $A$ of $J_0(A)$ in the solution (\ref{eq:nowind1}) is a nonlinear 
solution of the wave eikonal equation (\ref{eq:3.23}).  

\begin{proposition}
The argument $A$ of the Green's function (\ref{eq:nowind1})-(\ref{eq:nowind2}) 
is a nonlinear solution of the wave eikonal equation (\ref{eq:3.23}).
\end{proposition}
\begin{proof}
We set 
\begin{equation}
S^2=A^2=\left(k x+\tilde{\beta} t\right)^2+k^2 y^2, \label{eq:3.25}
\end{equation}
where
\begin{equation}
\tilde{\beta}(k)=\frac{\beta k}{k^2 +k_d^2}. \label{eq:3.26}
\end{equation}
Differentiation of (\ref{eq:3.25}) gives the equations:
\begin{equation}
S_t=\frac{(k x+\tilde{\beta} t)\tilde{\beta}}{S}, 
\quad S_x=\frac{k(kx+\tilde{\beta} t)}{S}, \quad 
S_y=\frac{k^2 y}{S}. \label{eq:3.27}
\end{equation} 
From (\ref{eq:3.27}) we obtain:
\begin{equation}
S_t-\frac{\tilde{\beta}}{k} S_x=0,\quad S_x^2+S_y^2=k^2. \label{eq:3.28}
\end{equation}
Combining (\ref{eq:3.28}) gives:
\begin{equation}
S_t-\frac{\beta S_x}{(S_x^2+S_y^2+k_d^2)}=0, \label{eq:3.29}
\end{equation}
which shows that (\ref{eq:3.25}) satisfies the wave eikonal equation (\ref{eq:3.23}). 
\end{proof}

\begin{remark}
We derive the solution (\ref{eq:3.25}) of the wave eikonal equation (\ref{eq:3.23}) 
in Appendix A, using the method of characteristics. 
 The wave eikonal equation (\ref{eq:3.23}) has a complete integral
\begin{equation}
S=a x+ b y+\frac{\beta a}{a^2+b^2+k_d^2} t, \label{eq:3.30}
\end{equation}
 (e.g. Sneddon (1957)). 
The envelope of the family of plane waves (\ref{eq:3.30}) is a solution of the wave eikonal 
equation. For example if we assume $b=b(a)$ then the solution of the equations $S'(a)=0$  
in principle gives a solution for $a$ as a function of $x$,$y$ and $t$. Substitution of 
this function $a=a(x,y,t)$ and $b=b(a)$ in the expression for $S$ in (\ref{eq:3.30}) 
then yields an envelope type solution for $S$ of the wave eikonal equation (\ref{eq:3.23}). 
The group velocity surface can also be regarded as an envelope solution of a family of plane wave 
solutions of the wave eikonal 
equation. 
\end{remark} 

\subsection{The Veronis Green's function form}

There are other forms of the Green's function $\psi_G$  
that are equivalent to (\ref{eq:3.15}) which are obtained by using a different 
integration order in (\ref{eq:3.5}).
Below we show that the Fourier form (\ref{eq:3.5}) of the Green's function can be reduced 
to the Green's function form obtained by Veronis (1958). 

By setting
\begin{equation}
s=-i\omega, \label{eq:v1}
\end{equation}
in the solution (\ref{eq:3.5}) for $\psi_G$ and noting if the complex $\omega$ plane contour is 
displaced $-i c$  then 
\begin{equation}
-ic-\infty<\omega<-ic+\infty \quad \implies\quad c-i\infty<s<c+i\infty, \label{eq:v2}
\end{equation}
then the Fourier solution (\ref{eq:3.5}) reduces to the Fourier-Laplace form:
\begin{equation}
\psi_G=-\frac{N}{4\pi^2}\int_{c-i\infty}^{c+i\infty} \frac{ds}{2\pi i}
\int_{-\infty}^{\infty} dk_x \int_{-\infty}^\infty dk_y 
\frac{\exp\left(st-i{\bf k}{\bf\cdot}{\bf x}\right)}
{s(k_x^2+k_y^2+k_d^2)+ik_x\beta} \label{eq:v3}
\end{equation}
The denominator in (\ref{eq:v3}) can be factored in the form:
\begin{equation}
D(s,{\bf k})=s\left(k_x^2+k_y^2+k_d^2\right)+i\beta k_x
\equiv s\left(k_x-k_x^-\right)\left(k_x-k_x^+\right), 
\label{eq:v4}
\end{equation}
where
\begin{equation}
k_x^\pm=i\left\{-\frac{\beta}{2s}\pm\left[\left(\frac{\beta}{2s}\right)^2+k_y^2+k_d^2\right]^{1/2}\right\}. 
\label{eq:v5}
\end{equation}
For real $s$, the pole $k_x=k_x^+$ is located in the upper half complex $k_x$-plane, 
and $k_x^-$ is located in the lower half $k_x$-plane. For $x>0$ we close the $k_x$-plane contour in 
the ${\rm Im}(k_x)<0$ half plane and for $x<0$ we close the contour in the ${\rm Im}(k_x)>0$ plane
(in order that the contour integrals around the large circular arcs closing the contours converge, 
and also to ensure $\psi_G\to 0$ as $|x|\to \infty$). Using Cauchy's residue theorem, we obtain:
\begin{equation}
\psi_G ({\bf x},t)=-N\int_{c-i\infty}^{c+i\infty} \frac{ds}{2\pi i} 
\int_{-\infty}^{\infty} \frac{dk_y}{2\pi} 
\frac{1}{2s} \exp (st-i k_y y) \frac{1}{\zeta} 
\exp\left(-\frac{\beta x}{2s}-\zeta |x|\right), \label{eq:v6}
\end{equation}
where
\begin{equation}
\zeta=\left[\left(\frac{\beta}{2s}\right)^2+k_y^2+k_d^2\right]^{1/2}. \label{eq:v7}
\end{equation}
The multi-valued function $\zeta$ is restricted to the first Riemann sheet where ${\rm Re}(\zeta)>0$. 
The result (\ref{eq:v6}) applies for both the $x>0$ and $x<0$ cases. 

Only the even part of the integrand in (\ref{eq:v6}) contributes to the $k_y$-integral. We 
obtain:
\begin{align}
\psi_G=&-N\int_{c-i\infty}^{c+i\infty} \frac{ds}{2\pi i} \frac{\exp(st)}{s}
\int_0^\infty \frac{dk_y}{2\pi} \cos(k_y y)
\exp\left(-\frac{\beta x}{2s}\right) \nonumber\\
&\ \times\frac{\exp\left[-\left[k_y^2+k_d^2+(\beta/2s)^2\right]^{1/2}
|x|\right]}
{\left[k_y^2+k_d^2+(\beta/2s)^2\right]^{1/2}}. \label{eq:v8}
\end{align}

Using the standard Fourier cosine transform:
\begin{equation}
\int_0^\infty \frac{1}{\sqrt{x^2+\alpha^2}}\exp\left[-\beta(x^2+\alpha^2)^{1/2}\right] \cos(xy) dx
=K_0\left[\alpha\left(\beta^2+y^2\right)^{1/2}\right], \label{eq:v9}
\end{equation}
(Erdelyi et al. (1954), Tables of Integral Transforms, vol. 1, p. 17, formula (27)], 
the Green's function (\ref{eq:v8}) reduces to the form:
\begin{equation}
\psi_G=-\frac{N}{2\pi}\int_{c-i\infty}^{c+i\infty} \frac{ds}{2\pi i} 
\exp\left(st-\frac{{\beta} x}{2s}\right)
\frac{1}{s} K_0\left(\left[k_d^2+\left(\frac{\beta}{2s}\right)^2\right]^{1/2} r\right), \label{eq:v10}
\end{equation}
where $r=(x^2+y^2)^{1/2}$ is cylindrical radius in the $xy$-plane. Equation (\ref{eq:v10}) is the form 
of the Rossby wave Green's function obtained by Veronis (1958). Veronis also gives further useful 
forms of the Green's function, and examples of application to oceanic Rossby waves. 

After a sequence of transformations, the Veronis Green's function (\ref{eq:v10}) 
for the case $k_d=0$ can be reduced to the form:
\begin{equation}
\psi_V:=\Psi_G=-\frac{N}{\pi} \int_0^\infty \frac{dz}{\sqrt{1+z^2}} 
J_0\left[2\sqrt{\alpha}(z^2+\gamma^2)^{1/2}\right], \label{eq:v11}
\end{equation}
where
\begin{equation}
\alpha=\beta r t,\quad x=r\cos\phi,\quad \gamma=\cos(\phi/2). \label{eq:v12}
\end{equation}
(see Appendix B, (\ref{eq:B20})). This is the final form of the Veronis (1958)  Green's function 
for $k_d=0$  
(his equation (21)). However, equation (21) of Veronis has $\alpha$ rather than $\sqrt{\alpha}$ 
given above in (\ref{eq:v11}). This is a typographical error in Veronis (1958), equation (21).  
 Using the properties of the Bessel function $J_0(z)$ it 
follows that $\psi_V$ satisfies the wave equation:
\begin{equation}
\frac{\partial^2\psi_V}{\partial\nu\partial\gamma} +4\nu\gamma \psi_V=0
\quad\hbox{where}\quad \nu=\sqrt{\alpha}. \label{eq:v13}
\end{equation}
We show in Appendix C, that $\psi_V$ in (\ref{eq:v11})  is equivalent to the Green's function 
in (\ref{eq:nowind1})-(\ref{eq:nowind2}) in proposition \ref{prop31} for the case $k_d=0$ 
(i.e. for an infinite Rossby deformation radius). 

The question arises: Is there a more general version 
of the Veronis solution (\ref{eq:v11})-(\ref{eq:v12})  
that applies for the more general case for a finite Rossby deformation radius $R_d$ (i.e. for 
$k_d\neq 0$)? Such a solution form for $\psi_G$ is given below. 
\begin{proposition}\label{prop33}
The Rossby wave Green's function (\ref{eq:nowind1})-(\ref{eq:nowind2}) can be written in a form,  
that depends on whether the observation point is located in (\romannumeral1)\ the region $r< \beta t/k_d^2$
(the near region) or (\romannumeral2)\ 
$r>\beta t/k_d^2$ (the far region). 
In the near region (\romannumeral1) the solution has the form:
\begin{align}
\psi_G=&-\frac{N}{2\pi}\left(\int_0^{z_{01}} \frac{dz}{\sqrt{z^2+1}} J_0(A_1) 
+\int_0^\infty \frac{dz}{\sqrt{z^2+1}} J_0(A_2)\right)\nonumber\\
\equiv&-\frac{N}{2\pi}\int_{-|z_{01}|}^\infty \frac{dz}{\sqrt{z^2+1}} J_0(A_2), \label{eq:wv1}
\end{align}
where
\begin{equation}
z=\frac{r|k^2-k_c^2|}{\sqrt{4\alpha}(k^2+k_d^2)^{1/2}}, \quad k_c^2=\frac{\beta t}{r}-k_d^2, \label{eq:wv2}
\end{equation}
defines the transformation from the integration variable $k$ in (\ref{eq:nowind1})-(\ref{eq:nowind2}) 
and the new integration variable $z=z(k)$. The transformation $z=z(k)$ is chosen to ensure that the 
argument of the Bessel function in solution (\ref{eq:nowind1})-(\ref{eq:nowind2}) 
remains invariant under the transformation (see Appendix B for the case $k_d=0$). 
We refer to the region $0<k<k_c$ as region 1, and $k>k_c$ 
as region 2 in $k$-space. The inverse of the transformation (\ref{eq:wv2}) is given by the equations:
\begin{equation}
(k^2+k_d^2)r^2=\biggl\{\begin{array}{ccc}
\alpha\left(\sqrt{z^2+1}-z\right)^2&\hbox{if}&0<k<k_c,\\
\alpha\left(\sqrt{z^2+1}+z\right)^2&\hbox{if}&k>k_c.
\end{array}
\biggr. \label{eq:wv3}
\end{equation}
The  quantities $A_j$ in (\ref{eq:wv1}) are defined by the equations: 
\begin{equation}  
A_j=\left[4\hat{\alpha}_j(z^2+\gamma^2)\right]^{1/2},\quad \hat{\alpha}_j=\alpha\delta_j, \quad 
\delta =\frac{k^2}{k^2+k_d^2},\quad  j=1,2 \label{eq:wv4}
\end{equation}
 where the index $j$ refers to whether $0<k<k_c$  or $k>k_c$. The parameters $\delta_j$ ($j=1,2$) 
and $z_{01}$ are defined as:
\begin{align}
 z_{01}=& \frac{k_c^2 r^2}{\sqrt{4\alpha} k_d r}
\equiv \frac{\alpha-k_d^2 r^2}{\sqrt{4\alpha} k_d r}\equiv z(k=0), \nonumber\\
\delta_1=&1-\frac{k_d^2 r^2}{\alpha\left[\sqrt{z^2+1}-z\right]^2}, \quad
\delta_2=1-\frac{k_d^2 r^2}{\alpha\left[\sqrt{z^2+1}+z\right]^2}\equiv \delta_1(-z). \label{eq:wv5}
\end{align}
In (\ref{eq:wv5}), $z_{01}$ refers to the value of $z(k)$ at $k=0$, where $z(k)$ is given by (\ref{eq:wv2}). 
  
In the far region, $r>\beta t/k_d^2$, the solution for $\psi_G$ is given by:
\begin{equation}
\psi_G=-\frac{N}{2\pi}\int_{z_{02}}^\infty \frac{dz}{\sqrt{z^2+1}} J_0(A), \label{eq:wv6}
\end{equation}
where
\begin{align}
&z=\frac{r(k^2+k_d^2-\beta t/r)}{\sqrt{4\alpha} (k^2+k_d^2)^{1/2}}, \quad 
z_{02}= \frac{(k_d^2-\beta t/r)r^2}{\sqrt{4\alpha} k_d r}, \nonumber\\
&(k^2+k_d^2)r^2=\alpha\left(\sqrt{z^2+1}+z\right)^2, \label{eq:wv7}
\end{align}
defines the $z=z(k)$ and $k=k(z)$ transformations and $z_{02}=z(0)$, where $z(k)$ is given by 
(\ref{eq:wv7}).  
 The argument $A$ of $J_0(A)$ is given by:
\begin{align}
A=&\left[4\hat{\alpha}(z^2+\gamma^2)\right]^{1/2}, \quad \hat{\alpha}=\alpha\delta, \nonumber\\
\delta=&\frac{k^2}{k^2+k_d^2}\equiv 1-\frac{k_d^2 r^2}{\alpha(\sqrt{z^2+1}+z)^2}. \label{eq:wv8}
\end{align}
\end{proposition}
\begin{remark}
The solution for $k_d=0$ in (\ref{eq:v11})-(\ref{eq:v12}) follows by letting $k_d\to 0$ 
in the near field solution (\ref{eq:wv1}). In that limit, $z_{01}\to\infty$ and there is no far 
field solution in that case.  
\end{remark}
\begin{proof}
Below we derive the transformation (\ref{eq:wv2}). Based on the Veronis (1958) 
Green's function form (\ref{eq:v11}) for $k_d=0$ and the Green's function (\ref{eq:nowind1})-(\ref{eq:nowind2}) 
for $k_d\neq 0$, we set 
\begin{equation}
A^2=4\hat\alpha(z^2+\gamma^2)=k^2 r^2+2\hat{\beta} tx+\frac{\hat{\beta}^2t^2}{k^2}, \label{eq:wv9}
\end{equation}
where $A$ is the $J_0$ Bessel function argument in (\ref{eq:nowind1})-(\ref{eq:nowind2}), 
and we use the notation:
\begin{equation}
\hat{\alpha}=\alpha \frac{k^2}{k^2+k_d^2},\quad \hat{\beta}=\beta\frac{k^2}{k^2+k_d^2}. \label{eq:wv10}
\end{equation}
Noting that 
\begin{equation}
4\hat{\alpha}\gamma^2=2\hat{\beta}t r +2\hat{\beta}t x, \label{eq:wv11}
\end{equation}
we obtain:
\begin{equation}
4\hat{\alpha} z^2=\left(kr- \frac{\hat{\alpha}}{kr}\right)^2. \label{eq:wv12}
\end{equation}
Solving (\ref{eq:wv12}) for $z$ gives the formula (\ref{eq:wv2}) for $z$, where we require that $z>0$. 
To obtain the inverse transformation (\ref{eq:wv3}) note that 
\begin{equation}
\mu=(k^2+k_d^2)r^2, \label{eq:wv13}
\end{equation}
satisfies the quadratic equation:
\begin{equation}
\mu^2-2\alpha (1+2z^2)\mu+\alpha^2=0, \label{eq:wv14}
\end{equation}
which leads to the inverse transformation (\ref{eq:wv3}). The remainder of the proof is straightforward.  
\end{proof}

\subsubsection{Asymptotics for large $t$}
At large time $t$ the behaviour of $\psi_G$ is correlated with the behavior of its 
Laplace transform as $s\to 0$ in (\ref{eq:v10}). In the limit as $s\to 0$, $\beta/s\to\infty$ and 
the argument of $K_0(\zeta)$, $\zeta\to\infty$, i.e. we can use the asymptotic 
form of $K_0(\zeta)$ for large $\zeta$, namely
\begin{equation}
K_0(\zeta)\sim \left(\frac{\pi}{2\zeta}\right)^{1/2} \exp(-\zeta)\quad\hbox{as}\quad \zeta\to\infty, 
\label{eq:v14}
\end{equation}
(Abramowitz and Stegun (1965), formula 9.7.2, p. 378). Using the approximation (\ref{eq:v11}) we obtain:
\begin{equation}
\psi_G\approx -\frac{N}{2\pi}\int_{c-i\infty}^{c+i\infty} \frac{ds}{2\pi i} \left(\frac{\pi}{\beta s r}\right)^{1/2} 
\exp\left[st- \frac{\beta}{2s}\left(x+r\right)\right], \label{eq:v15}
\end{equation}
as $t\to\infty$. 

Using the inverse Laplace transform:
\begin{equation}
\frac{1}{2\pi i}\int_{c-i\infty}^{c+i\infty} 
\exp(pt) p^{-1/2}\exp\left(-\frac{\alpha}{p}\right) dp
= \frac{1}{\sqrt{\pi t}} \cos\left(2\sqrt{\alpha}\sqrt{t}\right), \label{eq:v16}
\end{equation}
(Erdelyi et al. (1954), Vol. 1, p. 245, formula (37)), we obtain
\begin{equation}
\psi_G\sim -\frac{N}{2\pi}\frac{1}{\sqrt{\beta t r}}
\cos\left(\left[2\beta t(x+r)\right]^{1/2}\right), \label{eq:v17}
\end{equation}
for the form of $\psi_G$ at large time $t$. Equation (\ref{eq:v17}) shows that 
$\psi_G$ is maximal at locations where
\begin{equation}
2\beta t(x+r)=4 n^2 \pi^2\quad (n=\hbox{integer}). \label{eq:v18}
\end{equation}

\begin{proposition}
The phase:
\begin{equation}
S=\left[2\beta t(x+r)\right]^{1/2}, \label{eq:v19}
\end{equation}
of the asymptotic solution (\ref{eq:v17}) for $\psi_G$ satisfies the wave eikonal 
equation (\ref{eq:3.23}) with $k_d=0$. 
\end{proposition}

\begin{proof}
Differentiation of (\ref{eq:v19}) gives:
\begin{equation}
S_t=\frac{\beta(x+r)}{S},\quad S_x=\frac{\beta t (x+r)}{rS},\quad S_y=\frac{\beta y t}{rS}. 
\label{eq:v20}
\end{equation}
Using (\ref{eq:v20}) we obtain:
\begin{equation}
S_x^2+S_y^2=\frac{\beta t}{r}. \label{eq:v21}
\end{equation}
Thus, 
\begin{equation}
(S_x^2+S_y^2)S_t-\beta S_x=\frac{\beta t}{r}\frac{\beta (x+r)}{S}-\beta \frac{\beta t(x+r)}{rS}=0. 
\label{eq:v22}
\end{equation}
Thus, $S$ in (\ref{eq:v19}) satisfies the wave eikonal equation (\ref{eq:3.23}) with 
$k_d=0$. 
\end{proof}

\subsubsection{Asymptotics as $t\to 0$}
In this section we discuss the asymptotic form of the Veronis (1958) Rossby 
wave Green's function (\ref{eq:v10}) as $t\to 0$.
\begin{proposition}
For $k_d=0$ (infinite Rossby deformation radius), the Rossby wave Green's function (\ref{eq:v10}) 
in the limit as $t\to 0$ has the form:
\begin{equation}
\psi_G\sim \frac{N}{2\pi}\biggl\{\left[\gamma-\ln\left(\frac{\gamma\beta r t}{4}\right)\right]
+\frac{\beta xt}{2}\left[1-\gamma-\ln\left(\frac{\gamma\beta rt}{4}\right)\right]+\ldots \biggr\}H(t), 
\label{eq:v23}
\end{equation}
where $H(t)$ is the Heaviside step function and
\begin{equation}
\gamma=\lim_{n\to\infty}\left(\sum_{m=1}^n\frac{1}{m}-\ln(n)\right)=0.57721566\ldots, \label{eq:v24}
\end{equation}
is the Euler-Mascheroni constant (e.g. Abramowitz and Stegun (1965), formula 4.1.32, p. 68).
\end{proposition}

\begin{proof}
For $k_d=0$ solution (\ref{eq:v10}) for $\psi_G$ has the simpler form:
\begin{equation}
\psi_G=-\frac{N}{2\pi} \int_{c-i\infty}^{c+i\infty} \frac{ds}{2\pi i} 
\exp\left(st-\frac{\beta x}{2s}\right) \frac{1}{s} K_0\left(\frac{\beta r}{2 s}\right). \label{eq:v25}
\end{equation}
The solution for $\psi_G$ as $t\to 0$ is associated with $s\to\infty$ in Laplace transform 
space. As $s\to\infty$ the argument of the $K_0$ Bessel function, 
$\zeta=\beta r/(2s)\to 0$ as $s\to\infty$. 
Using the expansion for $K_0(\zeta)$ for small $\zeta$:
\begin{equation}
K_0(\zeta)\approx -\left[\gamma+\ln\left(\frac{\zeta}{2}\right)\right]+O(\zeta^2/4). \label{eq:v26}
\end{equation}
where $\gamma$ is the Euler-Mascheroni constant, we obtain the approximation:
\begin{equation}
\bar{\psi}_G\approx \frac{N}{2\pi} \left(1-\frac{\beta x}{2s}\right)\frac{1}{s}
\left[\gamma+\ln\left(\frac{\beta r}{4s}\right)\right], \label{eq:v27}
\end{equation}
for the Laplace transform of $\psi_G$ as $s\to\infty$. Using the inverse Laplace transforms:
\begin{equation}
{\cal L}^{-1}\left(s^{-1}\ln(s)\right)=-\ln(\gamma t),\quad 
{\cal L}^{-1}\left(s^{-2}\ln(s)\right)= t\left[1-\ln(\gamma t)\right], \label{eq:v28}
\end{equation}
(Erdelyi et al. (1954), Vol. 1, p. 250, formulas (1) and (2)), we obtain the solution (\ref{eq:v23}) for $\psi_G$
as $t\to 0$.
\end{proof}

\begin{proposition}
For $k_d\neq 0$ (finite Rossby deformation radius), the Rossby wave Green's function (\ref{eq:v10}) 
in the limit as $t\to 0$ (with $x$ and $y$ fixed), has the form:
\begin{align}
\psi_G\approx&-\frac{N}{2\pi}\biggl[K_0(k_d r)\left(J_0(\sqrt{2\beta x t})H(x)
+I_0(\sqrt{2\beta |x| t})[1-H(x)]\right)\nonumber\\
&\quad -K_1(k_d r) \frac{\beta^2r^2 t^2}{16k_d r}\biggr], \label{eq:v29}
\end{align}
where $I_n(z)$ and $K_n(z)$ are modified Bessel functions of the first and second kind, 
and $H(x)$ is the Heaviside step function. 
\end{proposition}
\begin{remark}
The arguments of the Bessel functions $I_0(z)$ and $J_0(z)$ in (\ref{eq:v29}) are necessarily small.
\end{remark}

\begin{proof}
The solution (\ref{eq:v10}) for $\psi_G$ has the inverse Laplace transform:
\begin{equation}
\bar{\psi}_G=-\frac{N}{2\pi} \exp\left(-\frac{\beta x}{2s}\right) \frac{1}{s} 
K_0\left[\left(k_d^2+\left(\frac{\beta}{2s}\right)^2\right)^{1/2} r\right]. \label{eq:v30}
\end{equation}
In the limit as $|s|\to \infty$ (\ref{eq:v30}) may be approximated by:
\begin{align}
\bar{\psi}_G\approx&-\frac{N}{2\pi}\exp\left(-\frac{\beta x}{2s}\right) \frac{1}{s}
\left[K_0(k_d r)+K_0'(k_d r)\frac{\beta^2 r}{8 s^2 k_d}\right]\nonumber\\
\approx &-\frac{N}{2\pi}\left[\exp\left(-\frac{\beta x}{2s}\right) \frac{1}{s}K_0(k_d r)
-K_1(k_d r)\frac{\beta^2 r}{8 k_d s^3}\right]. \label{eq:v31}
\end{align}
Laplace inversion of (\ref{eq:v31}) gives the approximation (\ref{eq:v29}) for $\psi_G$ as $t\to 0$. 
In the inversion of (\ref{eq:v31}) we used inverse Laplace transforms from Erdelyi et al. (1954), 
vol. 1, p. 245, formulas (35) and (40).
\end{proof}
\begin{remark}
The approximation (\ref{eq:v29}) shows that $\psi_G$ is bounded for $r\neq 0$. However as $r\to 0$
it becomes unbounded. The wavenumber $k_d$ (the Rossby inverse deformation radius) plays an important role 
near $r=0$. The result for $k_d=0$ in (\ref{eq:v23}) is clearly different than for 
$k_d\neq 0$ given in (\ref{eq:v29}).
\end{remark}

\section{Green's function with wind ($\Omega\neq 0$)}
\begin{proposition}
The Green's function solution $\psi_G$ of the Rossby wave equation (\ref{eq:2.6}) with delta function 
source term (\ref{eq:2.7}) for an azimuthal wind, with constant angular velocity $\Omega\neq 0$ 
is given by the integral:
\begin{equation}
\psi_G=-\frac{N}{2\pi} \int_0^\infty dk \frac{k}{k^2+k_d^2} J_0(D), \label{eq:rotate1}
\end{equation}
where
\begin{align}
D=&\left[a^2+D_1^2-2a D_1\cos(\Omega t+\epsilon)\right]^{1/2}, \nonumber\\
a=&\frac{k \beta}{(k^2+k_d^2)\Omega},\quad D_1=\left[(k x)^2+\left(k y+a\right)^2\right]^{1/2}, 
\nonumber\\
\sin\epsilon=&\frac{k x}{D_1},\quad \cos\epsilon=\frac{(k y+a)}{D_1}. \label{eq:rotate2}
\end{align}
and $J_0(D)$ is a Bessel function of the first kind of order zero and argument $D$. An alternative form for $D$ in $(x,y,t)$ coordinates is:
\begin{equation}
D=\left\{\left[k x+a\sin(\Omega t)\right]^2
+\left[k y+a-a\cos(\Omega t)\right]^2\right\}^{1/2}. 
\label{eq:rotate3}
\end{equation}
In the limit as $\Omega\to 0$ 
the Green's function (\ref{eq:rotate1}) reduces to the Rossby wave Green's function (\ref{eq:nowind1}).
\end{proposition}

\begin{proof}
 As in the solution method for $\Omega=0$ case (Section 3), we use the Fourier transform 
$\bar{\psi}({\bf k},\omega)$ defined by (\ref{eq:3.1}). The Fourier transforms of the various derivative terms in (\ref{eq:2.6}) are:
\begin{align}
{\cal F}\left[\derv{t}\left(\psi_{xx}+\psi_{yy}\right)\right]=&-i\omega k^2\bar{\psi}, 
\quad {\cal F}\left(\beta\psi_x\right)=-ik_x\beta\bar{\psi}, \nonumber\\
{\cal F}\left(x h_y\right)=&-k_y\deriv{\bar{h}}{k_x},
\quad {\cal F}\left(y h_x\right)=-k_x\deriv{\bar{h}}{k_y}, \label{eq:4.1}
\end{align}
where in the application of interest:
\begin{equation}
h=\psi_{xx}+\psi_{yy},\quad \bar{h}=-k^2 \bar{\psi}. \label{eq:4.2}
\end{equation}
Using (\ref{eq:4.1}) and (\ref{eq:4.2}) we obtain:
\begin{align}
\Omega{\cal F}\left[\left(-y\derv{x}+x\derv{y}\right)\nabla_\perp^2\psi\right]
=&\Omega\left(k_x\derv{k_y}-k_y\derv{k_x}\right)\left(-k^2\bar{\psi}\right)
=-\Omega\derv{\Phi}\left(k^2\bar{\psi}\right),\nonumber\\
\Omega{\cal F}\left[\left(-y\derv{x}+x\derv{y}\right)(-k_d^2\psi)\right]
=&\Omega\left(k_x\derv{k_y}-k_y\derv{k_x}\right)\left(-k_d^2\bar{\psi}\right)
=-\Omega\derv{\Phi}\left(k_d^2\bar{\psi}\right) \label{eq:4.3}
\end{align}
In (\ref{eq:4.3}) we have used the result:
\begin{equation}
k_x\derv{k_y}-k_y\derv{k_x}=\derv{\Phi}, \label{eq:4.4}
\end{equation}
where $(k_x,k_y)=k\left(\cos\Phi,\sin\Phi\right)$ is the ${\bf k}$-vector in cylindrical 
polar coordinates. 

Taking the Fourier transform of the Rossby wave equation (\ref{eq:2.6}), (\ref{eq:2.6}) reduces 
to an ordinary differential equation in transform space:
\begin{equation}
-i\omega (k^2+k_d^2) \bar{\psi}
-i k_x\beta\bar{\psi}-\Omega \frac{d}{d\Phi}\left((k^2+k_d^2)\bar{\psi}\right)=N. 
\label{eq:4.5}
\end{equation}
Setting
\begin{equation}
\bar{\omega}=\frac{\omega}{\Omega},\quad \frac{k_x\beta}{(k^2+k_d^2)\Omega}=a\cos\Phi, \quad 
a=\frac{k\beta}{(k^2+k_d^2)\Omega}, \label{eq:4.6}
\end{equation}
 reduces  (\ref{eq:4.5}) to the equation:
\begin{equation}
\frac{d\bar{\psi}}{d\Phi}+i\left(\bar{\omega}+a\cos\Phi\right) \bar{\psi}
=-\frac{N}{\Omega (k^2+k_d^2)}. \label{eq:4.7}
\end{equation}
The integrating factor for (\ref{eq:4.7}) is:
\begin{equation}
I(\Phi)= \exp\left[i\int^\Phi \left(\bar{\omega}+a\cos\Phi\right)d\Phi\right] 
=\exp\left[i\left(\bar{\omega}\Phi+a\sin\Phi\right)\right]. \label{eq:4.8}
\end{equation}

Thus, (\ref{eq:4.7}) can be written in the form:
\begin{equation}
\frac{d}{d\Phi}\left(I\bar{\psi}\right)=-\frac{N I(\Phi)}{\Omega (k^2+k_d^2)}. \label{eq:4.9}
\end{equation}
Equation (\ref{eq:4.9}) can be integrated to give the solution for $\bar{\psi}(\Phi)$ in the form:
\begin{equation}
I(\Phi)\bar{\psi}(\Phi)=I(0)\bar{\psi}(0)-\frac{N}{\Omega (k^2+k_d^2)} \mathscr{I}(\Phi), \label{eq:4.10}
\end{equation}
where
\begin{equation}
\mathscr{I}(\Phi)=\int_0^\Phi I(\Phi')d\Phi'. \label{eq:4.11}
\end{equation}
Since $\Phi$ is the azimuthal angle in ${\bf k}$-space, $\bar{\psi}(\Phi)$ must be periodic in $\Phi$ 
with period $2\pi$, i.e. 
\begin{equation}
\bar{\psi}(2\pi)=\bar{\psi}(0). \label{eq:4.12}
\end{equation}
setting $\Phi=2\pi$ in (\ref{eq:4.10}) and using the periodicity condition (\ref{eq:4.12}) we obtain the 
constraint equation:
\begin{equation}
\bar{\psi}(0)=-\frac{N}{\Omega (k^2+k_d^2)} \frac{\mathscr{I}(2\pi)}{[I(2\pi)-I(0)]}, \label{eq:4.13}
\end{equation}
where
\begin{equation}
I(2\pi)=\exp(2\pi i\bar{\omega})\quad \hbox{and}\quad I(0)=1. \label{eq:4.14}
\end{equation}
Thus, the formal solution for $\bar{\psi}$ can now be obtained by 
substituting (\ref{eq:4.13}) and (\ref{eq:4.14})  in (\ref{eq:4.10}) to obtain the solution for 
$\bar{\psi}(\Phi)$ in the form:
\begin{align}
I(\Phi)\bar{\psi}(\Phi)=&\frac{N}{\Omega (k^2+k_d^2)[I(2\pi)-I(0)]}\biggl\{-\mathscr{I}(2\pi) 
-[I(2\pi)-I(0)]\mathscr{I}(\Phi)\biggr\}\nonumber\\
=&\frac{N}{\Omega (k^2+k_d^2)[1-\exp(2\pi i\bar{\omega})]}
\left[ \int_\Phi^{2\pi} I(\Phi')d\Phi'+\exp(2\pi i\bar{\omega})\int_0^\Phi I(\Phi')d\Phi'\right]. 
\label{eq:4.15}
\end{align}

The solution (\ref{eq:4.10}) or (\ref{eq:4.15}) can be related to Bessel function expansions, 
by using the Bessel function generating identity (\ref{eq:3.14}). In particular,
\begin{align}
\mathscr{I}(2\pi)=&\int_0^{2\pi} \exp\left(i\bar{\omega}\Phi+ia\sin\Phi\right) d\Phi\nonumber\\
=&\int_0^{2\pi}\exp(i\bar{\omega}\Phi)\left(\sum_{n=-\infty}^\infty J_n(a)\exp(in\Phi)\right) d\Phi. \label{eq:4.16}
\end{align}
For general $\bar{\omega}$ with $\bar{\omega}+n\neq 0$, (\ref{eq:4.16}) gives:
\begin{equation}
\mathscr{I}(2\pi)=\sum_{n=-\infty}^\infty \frac{[\exp(2\pi i\bar{\omega})-1]}{i(\bar{\omega}+n)} J_n(a). 
\label{eq:4.17}
\end{equation}
Using (\ref{eq:4.17}) in (\ref{eq:4.13}) gives the result:
\begin{equation}
\bar{\psi}(0)=\frac{i N}{\Omega (k^2+k_d^2)} \sum_{n=-\infty}^\infty \frac{J_n(a)}{(\bar\omega+n)}, 
\label{eq:4.18}
\end{equation}
for $\bar{\psi}(0)$. Similarly, 
\begin{equation}
\mathscr{I}(\Phi)=-i\sum_{n=-\infty}^\infty \frac{J_n(a)}{(\bar\omega+n)}
\left(\exp[i(\bar{\omega}+n)\Phi]-1\right). \label{eq:4.19}
\end{equation}
Substitute (\ref{eq:4.17})-(\ref{eq:4.19}) in (\ref{eq:4.10}) then gives the solution for $\bar{\psi}(\Phi)$ 
in the form:
\begin{equation}
\bar{\psi}(\Phi)=\frac{i N}{\Omega (k^2+k_d^2)}\exp(-ia\sin\Phi) 
\left(\sum_{n=-\infty}^\infty\frac{J_n(a)}{(\bar\omega+n)}\exp(in\Phi)\right). \label{eq:4.20}
\end{equation}
Note that $\bar{\psi}(\Phi+2\pi)=\bar{\psi}(\Phi)$ in (\ref{eq:4.20}). 

Substitute the solution (\ref{eq:4.20}) for $\bar{\psi}(\Phi)$ in the inverse Fourier transform formula:
\begin{align}
\psi(x,y,t)=&\frac{N}{8\pi^3} \int_{-\infty}^{\infty}  dk_x\int_{-\infty}^\infty dk_y 
\int_{-\infty}^\infty d\omega \exp[i(\omega t-{\bf k\cdot x})] 
\bar{\psi}({\bf k},\omega)\nonumber\\
\equiv& \frac{N}{8\pi^3} \int_0^{2\pi} d\Phi \int_0^{\infty} k dk 
\int_{-\infty}^{\infty} d\omega
\biggl\{\exp\left[ i\left(\omega t-k (x\cos\Phi+y\sin\Phi)\right)\right]\nonumber\\
&\times \frac{i}{\Omega (k^2+k_d^2)} \exp(-ia\sin\Phi)
\left(\sum_{n=-\infty}^{\infty} \frac{J_n(a)}{(\bar{\omega}+n)}\exp(i n\Phi)\right)\biggr\}. \label{eq:4.21}
\end{align}
gives the formal solution for $\psi(x,y,t)$. Writing
\begin{equation}
k x\cos\Phi+k y\sin\Phi+a\sin\Phi=D_1 \sin(\Phi+\epsilon), \label{eq:4.22}
\end{equation}
implies 
\begin{align}
k x=&D_1\sin\epsilon,\quad k y+a=D_1\cos\epsilon, \nonumber\\
D_1=&\left[(k x)^2+(k y +a)^2\right]^{1/2},\quad \tan\epsilon=\frac{k x}{(k y+a)}. 
\label{eq:4.23}
\end{align}
Writing the integrand in curly brackets in (\ref{eq:4.21}) as $\hat{\psi}$ we obtain:
\begin{align}
\hat{\psi}=&\frac{i\exp(i\omega t)}{\Omega (k^2+k_d^2)}\exp\left[-iD_1\sin(\Phi+\epsilon)\right] 
\sum_{n=-\infty}^\infty \frac{J_n(a)}{(\bar{\omega}+n)}\exp(i n\Phi)\nonumber\\
\equiv& \frac{i\exp(i\omega t)}{\Omega (k^2+k_d^2)}
\sum_{m=-\infty}^\infty J_m(D_1)\exp[-im(\Phi+\epsilon)]\sum_{n=-\infty}^\infty 
\frac{J_n(a)}{(\bar{\omega}+n)}\exp(i n\Phi)\nonumber\\
=&\frac{i\exp(i\omega t)}{\Omega (k^2+k_d^2)} \sum_{m=-\infty}^\infty \sum_{n=-\infty}^\infty 
\frac{J_m(D_1)J_n(a)}{(\bar{\omega}+n)}\exp[i(n-m)\Phi-i m\epsilon]. \label{eq:4.24}
\end{align}

The integral of $\hat{\psi}$ in (\ref{eq:4.24}) over $\Phi$ from $\Phi=0$ to $\Phi=2\pi$ gives:
\begin{equation}
\int_0^{2\pi}\hat{\psi}\ d\Phi=\frac{2\pi i\exp(i\omega t)}{\Omega (k^2+k_d^2)}\sum_{n=-\infty}^\infty 
\frac{J_n(D_1)J_n(a)}{(\bar{\omega}+n)}\exp(-i n\epsilon) \label{eq:4.25}
\end{equation}

Using the result (\ref{eq:4.25}) in (\ref{eq:4.21}) we obtain:
\begin{equation}
\psi(x,y,t)=\frac{i N}{4\pi^2}\int_0^\infty dk \frac{k}{k^2+k_d^2}
\int_{-\infty}^{\infty} d\omega 
\exp(i\omega t)\left(\sum_{n=-\infty}^\infty \frac{J_n(D_1)J_n(a)
\exp(-i n\epsilon)}{(\omega+n\Omega)}\right). \label{eq:4.26}
\end{equation}
Taking into account the poles of the integrand in the complex $\omega$-plane at $\omega=-n\Omega$, 
and deforming the contour near the poles with small semi-circular arcs 
$\omega=-n\Omega+r\exp(i\theta)$, ($-\pi\leq\theta\leq 0$) and completing the contour by a 
large semi-circular  arc $C_R:\ \omega=R\exp(i\theta)$ ($0\leq\theta\leq\pi$) in the $\rm{Im}(\omega)>0$ 
half plane, noting that $|\int_{C_R}|\to 0$ as $R\to\infty$, and using the Residue theorem, we obtain from 
(\ref{eq:4.26}) the solution form:
\begin{equation}
\psi(x,y,t)=-\frac{N}{2\pi}\int_0^\infty dk\ \frac{k}{k^2+k_d^2}\left(\sum_{n=-\infty}^\infty 
J_n(D_1)J_n(a)\exp[-in(\Omega t+\epsilon)]\right) H(t), \label{eq:4.27}
\end{equation}
where $H(t)$ is the Heaviside step function. 
Next we use the Neumann series identity:
\begin{equation}
\sum_{n=-\infty}^\infty J_n(D_1) J_n(a)\exp[-in(\Omega t+\epsilon)]=J_0(D), \label{eq:4.28}
\end{equation}
where
\begin{equation}
D=\left[a^2+D_1^2-2aD_1\cos(\Omega t+\epsilon)\right]^{1/2}. \label{eq:4.29}
\end{equation}
The result (\ref{eq:4.28})-(\ref{eq:4.29}) is a special case of the Neumann series formula (8.5.30) 
given by Gradshteyn and Ryzhik (2000), p. 930.
Using (\ref{eq:4.28})-(\ref{eq:4.29}) in (\ref{eq:4.27}) gives the Green's function (\ref{eq:rotate1}). 
This completes the proof.
\end{proof}

\subsection{\bf The limit as $\Omega\to 0$}
The Green's function (\ref{eq:rotate1}) for Rossby waves for $\Omega\neq 0$ becomes the Green's function 
(\ref{eq:nowind1}) in the limit as $\Omega\to 0$. To show this consider the form of $D$ in (\ref{eq:rotate3}) 
in the limit as $\Omega\to 0$. Using $a=\beta k/[(k^2+k_d^2)\Omega]$ and the 
approximations $\sin(\Omega t)\sim\Omega t$  and $\cos(\Omega t)\sim 1-\Omega^2 t^2/2$ for small $\Omega t$ in (\ref{eq:rotate3}) we obtain:
\begin{align}
D^2=&\left[\left(k x+a\sin(\Omega t)\right)^2
+\left(k y+a-a\cos(\Omega t)\right)^2\right]\nonumber\\
\approx&[k x+a\Omega t]^2+\left[k y+a\Omega^2t^2/2\right]^2. \label{eq:4.30}
\end{align}
However, 
\begin{equation}
a\Omega t=\frac{k\beta t}{k^2+k_d^2}, \quad \frac{a\Omega^2 t^2}{2}
=\frac{k\beta t}{k^2+k_d^2} \frac{\Omega t^2}{2}\to 0, 
\label{eq:4.30a}
\end{equation}
as $\Omega\to 0$. 
 Using these 
results for small $\Omega$, (\ref{eq:4.30}) gives $D^2\to A^2$ as $\Omega\to 0$. 
Thus, the solution (\ref{eq:rotate1}) 
reduces to the solution (\ref{eq:nowind1}) in the limit as $\Omega\to 0$. 

\begin{remark}
The Green's function (\ref{eq:rotate1}) can be written as:
\begin{equation}
\psi_G=-\frac{N}{2\pi}\int_0^\infty dk \frac{k}{k^2+k_d^2} \psi^D\quad \hbox{where}\quad \psi^D=J_0(D). 
\label{eq:4.31}
\end{equation}
Away from the source at $(x,y,t)=(0,0,0)$, $\psi^D$ satisfies the Rossby wave equation (\ref{eq:2.6}) 
with $Q=0$, i.e., 
\begin{equation}
\left[\derv{t}+\Omega\left(x\derv{y}-y\derv{x}\right)\right]
\left(\psi^D_{xx}+\psi^D_{yy}-k_d^2 \psi^D\right)+\beta\psi^D_x=0. \label{eq:4.32}
\end{equation}
This result follows by using the properties of the Bessel function $J_0(D)$ and the derivatives 
of the function $D(x,y,t)$ in (\ref{eq:rotate3}). 
\end{remark}

\section{Solution characteristics and examples}
In this section we investigate the physical characteristics of the Rossby wave Green's function 
solutions for the case of no rotating wind ($\Omega=0$) described by (\ref{eq:nowind1})-(\ref{eq:nowind2})
including both the cases of an infinite Rossby deformation radius ($k_d=0$) and also the rotating wind 
case ($\Omega\neq 0$) described by (\ref{eq:rotate1})-(\ref{eq:rotate3}). Section 5.1 discusses typical 
parameters for Rossby waves on the Earth.  The Veronis (1958) Green's 
function (both for the $k_d=0$ and $k_d\neq 0$ cases) is investigated in Section 5.2 
and the rotating wind Green's function (\ref{eq:rotate1})-(\ref{eq:rotate3}) ($\Omega\neq 0$) 
is studied in Section 5.3. Below we discuss the form of the fluid vorticity for the above 2 cases. 

\begin{proposition}
The vorticity for the Green's function (\ref{eq:nowind1})-(\ref{eq:nowind2}) 
with no rotating wind ($\Omega=0$) and for the rotating wind case ($\Omega\neq 0$) given by 
(\ref{eq:rotate1})-(\ref{eq:rotate3}) has local fluid vorticity:
\begin{equation}
\boldsymbol{\omega}=\nabla\times {\bf u}=(0,0,\zeta)^T\quad \hbox{where}\quad \zeta=\psi_{xx}+\psi_{yy}, 
\label{eq:5.1}
\end{equation}
and $\psi$ is the stream function. The vorticity $\zeta$ can be written in the form:
\begin{equation}
\zeta=\psi_{xx}+\psi_{yy}=\int_0^\infty dk \frac{k^3}{k^2+k_d^2} J_0(D), \label{eq:5.2}
\end{equation}
for the rotating wind case where $D$ is given by (\ref{eq:rotate3}). In the case of no rotating wind ($\Omega=0$)
$D$ in (\ref{eq:5.2}) is replaced by $A$ which is given by (\ref{eq:nowind2}). 
\end{proposition}

\begin{proof}
We give the proof for the case $\Omega\neq 0$ for the Green's function (\ref{eq:rotate1})-(\ref{eq:rotate3}).
A similar proof for the case $\Omega = 0$ applies for the no wind case. We omit the proof 
for the $\Omega=0$ case. 

For the case $\Omega\neq 0$ we write:
\begin{equation}
\psi_G=-\frac{N}{2\pi} \int_0^\infty dk\frac{k}{k^2+k_d^2} \tilde{\psi}\quad\hbox{where}
\quad \tilde{\psi}=J_0(D), \label{eq:5.3}
\end{equation} 
and $D$ is given by (\ref{eq:rotate3}). Assuming that it is valid to interchange 
the order of integration and differentiation in (\ref{eq:5.3}) we obtain:
\begin{equation}
\zeta=\nabla_\perp^2\psi_G=-\frac{N}{2\pi} \int_0^k dk\frac{k}{k^2+k_d^2} \nabla_\perp^2\tilde{\psi}. 
\label{eq:5.4}
\end{equation}
Using the definition of $\tilde{\psi}=J_0(D)$ we obtain:
\begin{equation}
\nabla_\perp^2 \tilde{\psi}=J_0'(D)\left(D_{xx}+D_{yy}\right)+J_0''(D)\left(D_x^2+D_y^2\right). 
\label{eq:5.5}
\end{equation}
Using (\ref{eq:rotate3}) for $D$ we obtain:
\begin{align}
D_x=&\frac{k[kx+a\sin(\Omega t)]}{D}, \quad D_y=\frac{k[ky+a-a \cos(\Omega t)]}{D}, \nonumber\\
D_{xx}=&\frac{k^2}{D^3}\left(D^2-[kx +a\sin(\Omega t)]^2\right), \nonumber\\
D_{yy}=&\frac{k^2}{D^3}\left(D^2-[ky+a-a\cos(\Omega t)]^2\right). \label{eq:5.6}
\end{align}
Using (\ref{eq:5.6}) results in the equations:
\begin{equation}
D_x^2+D_y^2=k^2,\quad D_{xx}+D_{yy}=\frac{k^2}{D}. \label{eq:5.7}
\end{equation}
Substitution of (\ref{eq:5.7}) into (\ref{eq:5.5}) gives:
\begin{equation}
\nabla_\perp^2\tilde{\psi}(D)= k^2\left[J_0''(D)+\frac{1}{D} J_0'(D)\right]=-k^2 J_0(D). 
\label{eq:5.8}
\end{equation}
In deriving (\ref{eq:5.8}) we used the fact that $J_0(D)$ satisfies Bessel's equation:
(e.g. Abramowitz and Stegun (1965), Ch. 9, formula (9.1.1), p. 358). 
Substitution of (\ref{eq:5.8}) for $\nabla_\perp^2\tilde{\psi}$ in (\ref{eq:5.4}) gives the result 
(\ref{eq:5.2}) for $\zeta$.  This completes the proof.
\end{proof}

\subsection{Typical Parameters}
Typical values for Rossby waves on the Earth are: $\Omega_E=2\pi/(24\times 3600)\sim 7.2722\times 10^{-5}$
 ${\rm sec}^{-1}$. The Rossby number, $Ro$ which is the ratio of inertial to Coriolis acceleration
$\sim (U^2/L)/(2\Omega U\sim U/(2\Omega L)$. 
For a typical pressure field in the troposphere, with $L\sim 1000\ {\rm km}$ 
and $U\sim 20\ {\rm m s}{}^{-1}$, the Rossby number  $Ro\sim 0.1375$ (Pedlosky 1987, p. 3). 
The Rossby wave parameter
$\beta$ is defined as $\beta=2\Omega_E\cos\theta_0/R_E$. Using $\bar{R}_E\sim 6371 {\rm km}$ we obtain
$\beta\sim 2.2829\times 10^{-11}\ {\rm m}^{-1}\ {\rm s}^{-1}\cos\theta_0$. Thus at the equator 
$\beta\sim  2.2829\times 10^{-11}\ {\rm m}^{-1}\ {\rm s}^{-1}$ and at $\theta_0=60^{\circ}$ latitude
$\beta\sim 1.1414\times 10^{-11}\ {\rm m}^{-1}\ {\rm s}^{-1}$. For $L\sim 1000\ {\rm km}$ and a wind speed $V_w=80\ 
{\rm km}\ {\rm hr}^{-1}$ at $r=L$, gives $\Omega=V_w/L=2\times 10^{-5}\ {\rm s}^{-1}$ for the parameter $\Omega$, 
which is comparable to $\Omega_E$.  

There are two estimates commonly used to calculate the Rossby deformation radius 
(e.g. Gill (1982),  Vallis (2006)). One can use the so-called barotropic Rossby 
deformation radius $R_{d1}=\sqrt{gH}/f_0$ where $g$ is the gravitational acceleration and 
$f_0=2\Omega_E\sin\theta_0$ is the Coriolis parameter. The average height of the troposphere $H=17$ km, 
and the mean gravitational scale height $h=8.5$ km. If one uses the mean height of the troposphere, 
in the formula for $R_{d1}$ one obtains $R_{d1}=5.6515\times 10^3\ {\rm km}$, 
but if one uses the gravitational scale height ($h=\sqrt{kT/mg}$) then $R_{d1}=3971$ km. 
In some cases, the baroclinic Rossby deformation radius is thought to be more appropriate (e.g. 
Vallis, 2006), which is given by the formula  $R_{d2}=N h/(n \pi f_0)$, where $N$ is the Brunt-V\"ais\"al\"a 
frequency $N=\sqrt{g(1/h-g/c_s^2)}$ where $c_s$ is the sound speed, and  $n$ is a positive integer. 
For the Earth's atmosphere $N/f_0\sim 100$ ($N\sim 0.01 {\rm s}^{-1}$, $f_0\sim 10^{-4} {\rm s}^{-1}$), 
$h\sim 10{\rm km}$, 
and $R_{d2}\sim 1000$ km. The largest value of $k_d$ from these estimates is $k_d=10^{-6} {\rm m}^{-1}$.  

\subsection{Non-rotating wind Green's function: $\Omega=0$}

\begin{figure}[htt]
\begin{center}
\includegraphics[width=10cm]{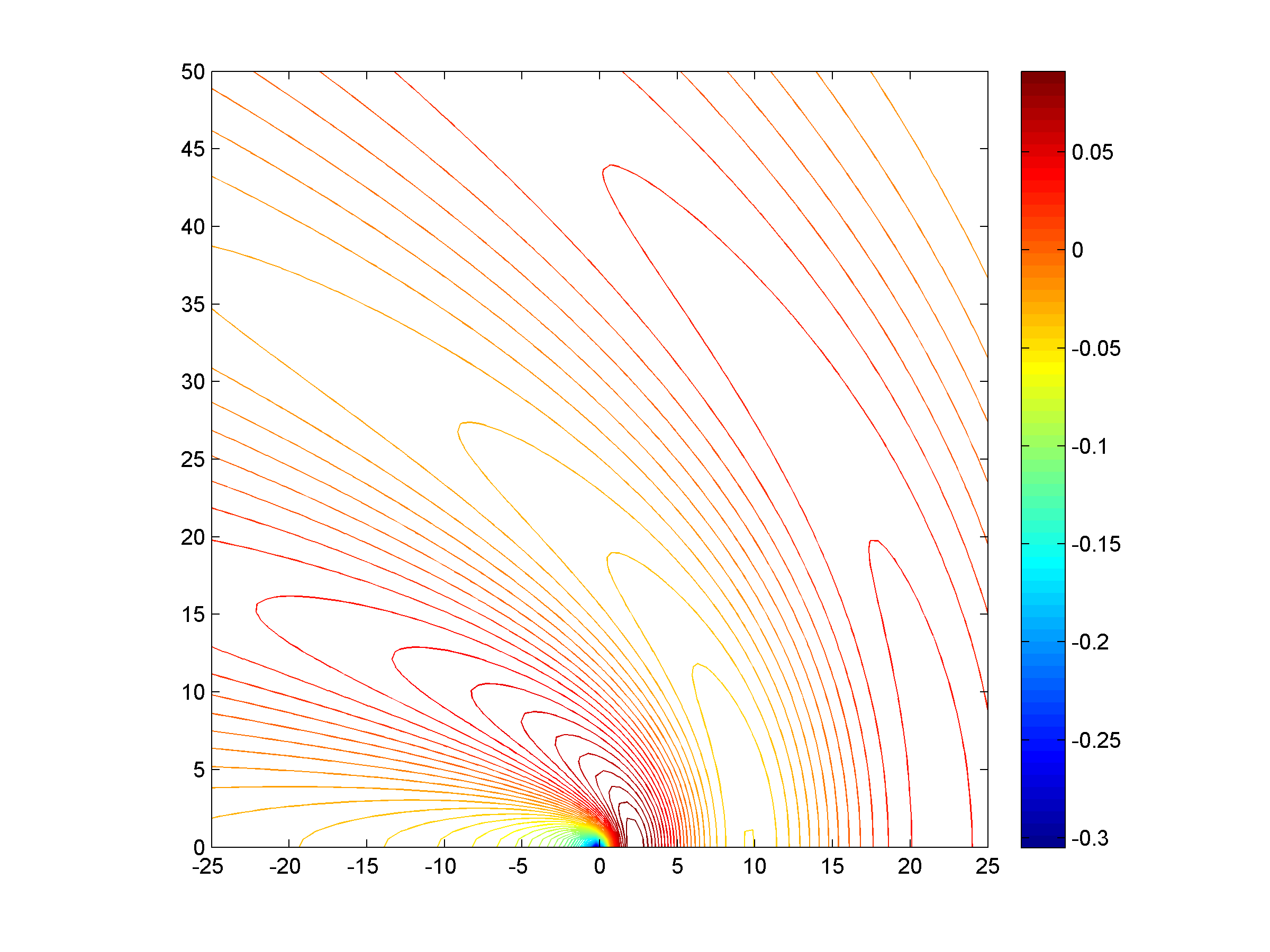}
\caption{Contour plots of the stream function $\psi_G$ in the $\beta$\hyp{}plane 
for the Veronis Green's function for Rossby waves (\ref{eq:nowind1})-(\ref{eq:nowind2}) 
with an infinite Rossby deformation radius $R_d$ ($k_d=0$ case).}\label{fig:rossby1}
\end{center}
\end{figure}

\begin{figure}[htt]
\begin{center}
\includegraphics[width=10cm]{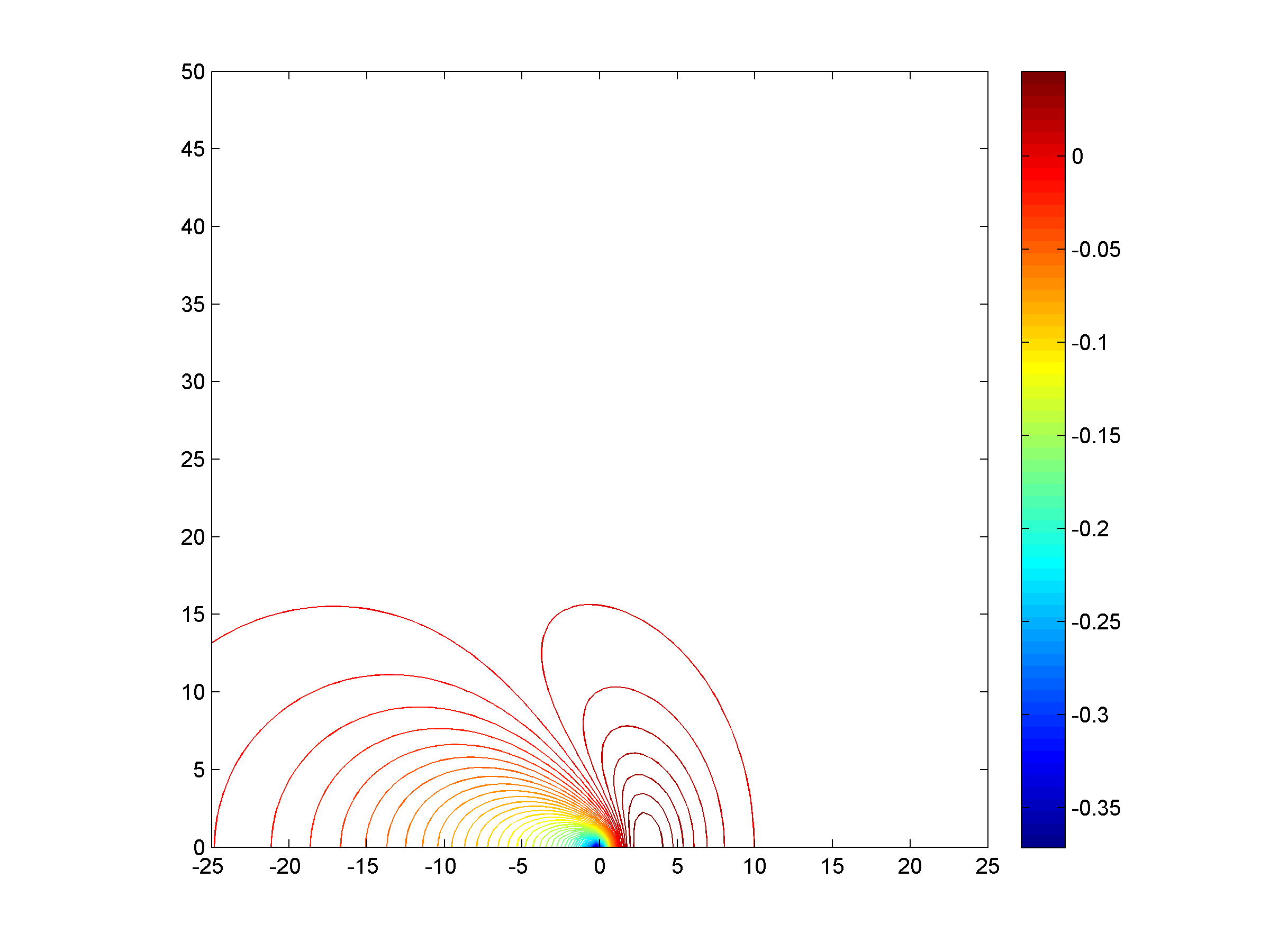}
\caption{Contour plots of the stream function $\psi_G$ in the $\beta$\hyp{}plane, 
for the Green's function (\ref{eq:nowind1})-(\ref{eq:nowind2}) for Rossby waves  
 with a finite Rossby deformation radius ($\bar{k}_d=0.25$).}\label{fig:rossby2} 
\end{center}
\end{figure}

\begin{figure}[htt]
\begin{center}
\includegraphics[width=10cm]{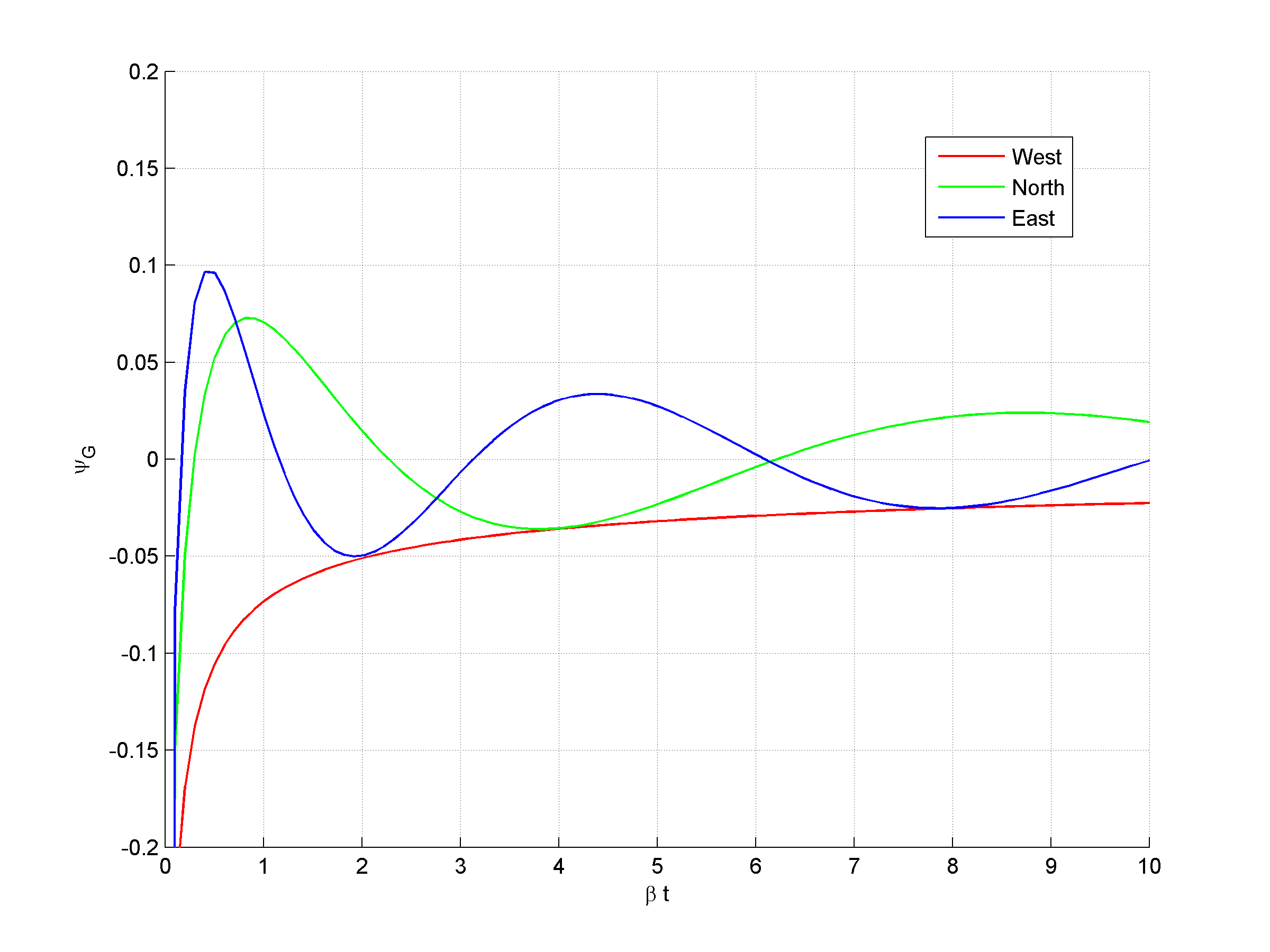}
\caption{Stream function $\psi_G$ versus $t$ for the 
 Green's function (\ref{eq:nowind1})-(\ref{eq:nowind2}) for Rossby waves
at $\bar{r}=5$  with an infinite Rossby deformation radius $k_d=0$. 
By symmetry, the southern solution (not shown) 
is the same as the northern solution.}\label{fig:rossby3}
\end{center}
\end{figure}

\begin{figure}[htt]
\begin{center}
\includegraphics[width=10cm]{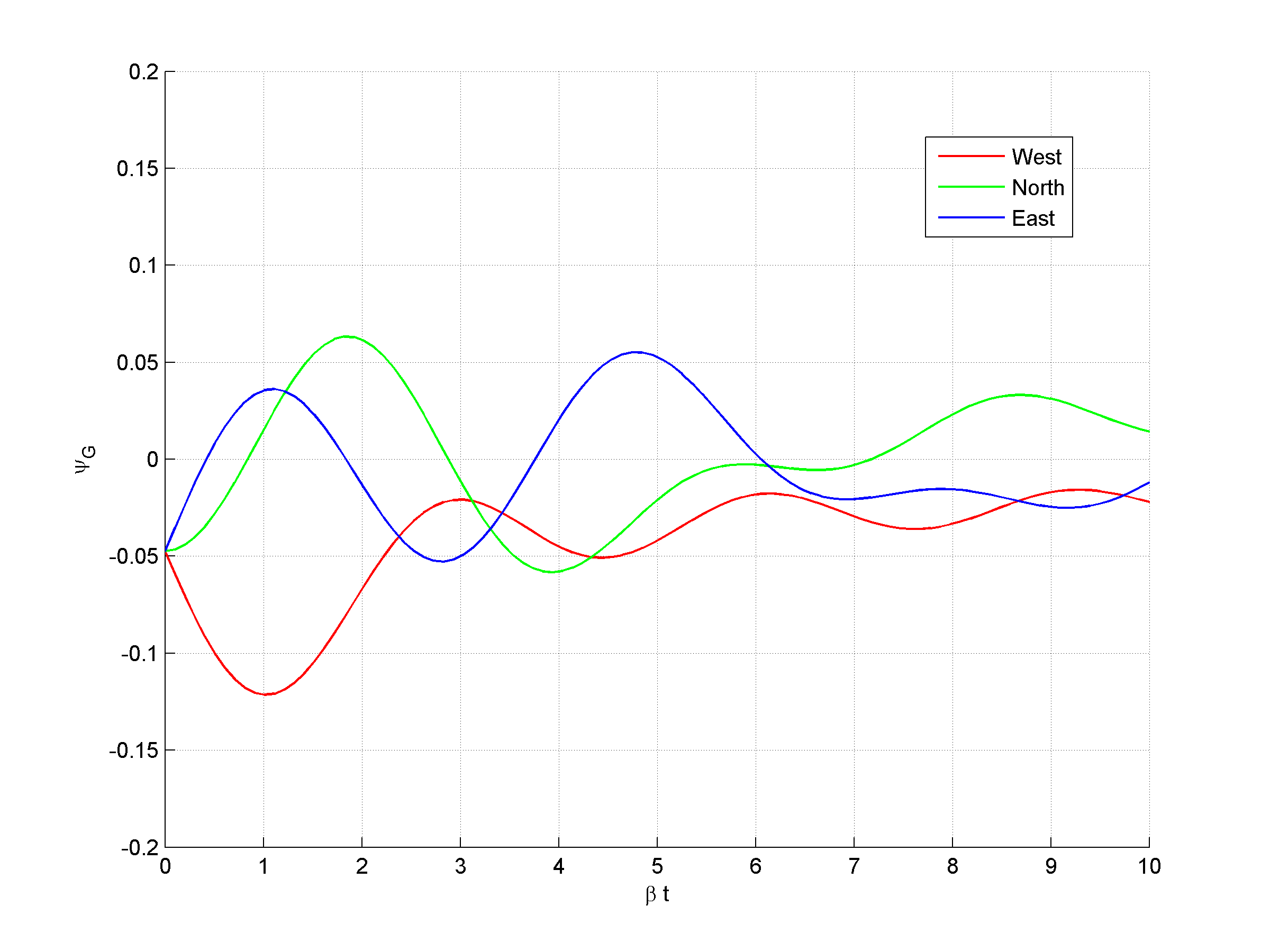}
\caption{Stream function $\psi_G$ versus $t$ for the 
 Green's function (\ref{eq:nowind1})-(\ref{eq:nowind2}) for Rossby waves
at $\bar{r}=5$, with a finite Rossby deformation radius ($\bar{k}_d=0.25$). By symmetry, 
the southern solution (not shown) is the same as the northern solution.}\label{fig:rossby4} 
\end{center}
\end{figure}

\begin{figure}[htt]
\begin{center}
\includegraphics[width=10cm]{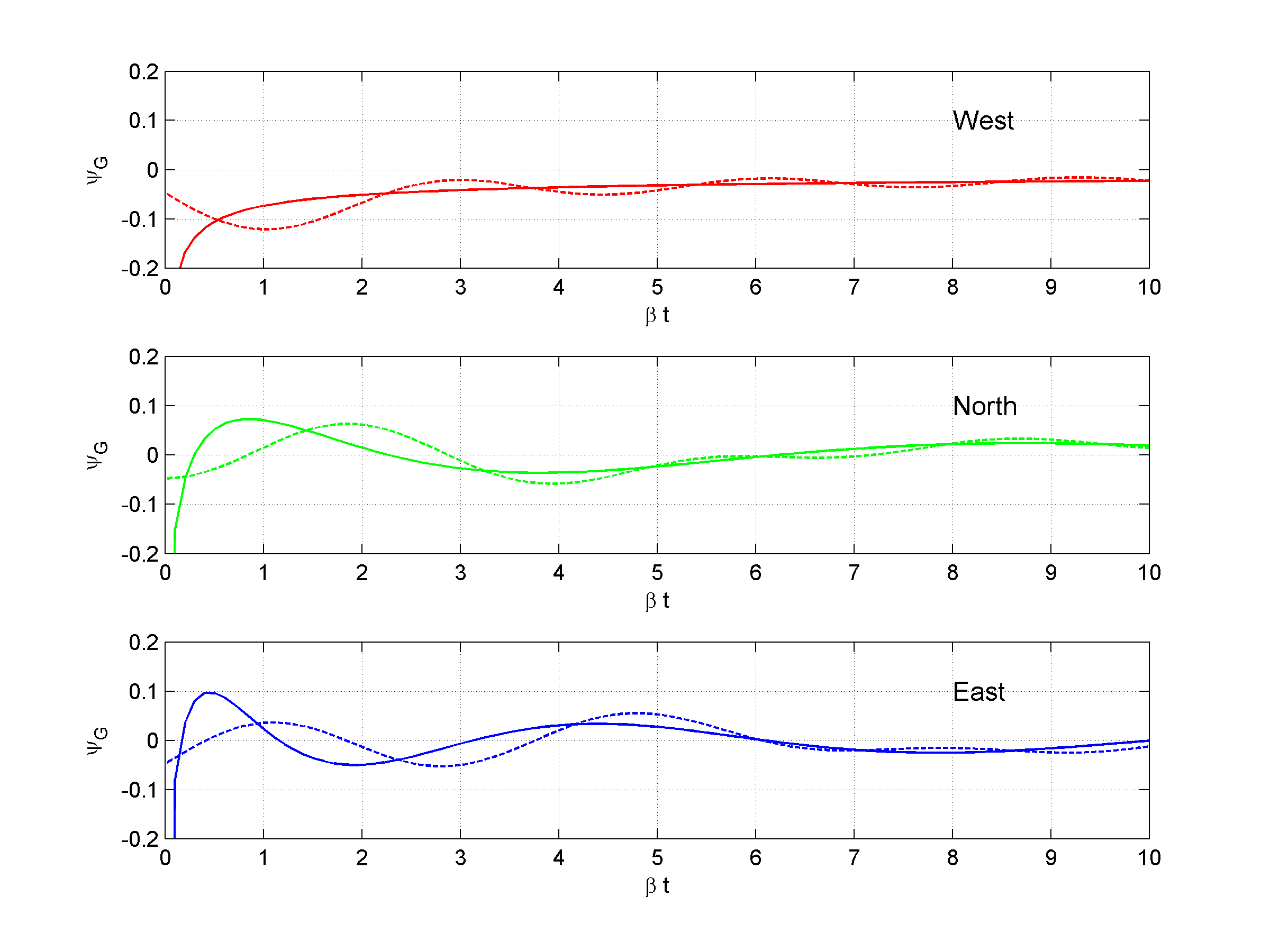}
\caption{Stream function $\psi_G$ versus $t$ for the
 Green's function (\ref{eq:nowind1})-(\ref{eq:nowind2}) for Rossby waves at $\bar{r}=5$
 with a finite Rossby deformation radius ($\bar{k}_d=0.25$ dashed curves) compared to the 
infinite Rossby radius solutions ($k_d=0$ solid curves).}\label{fig:rossby5} 
\end{center}
\end{figure}
\begin{figure}[ht]
\begin{center}
\includegraphics[width=16cm]{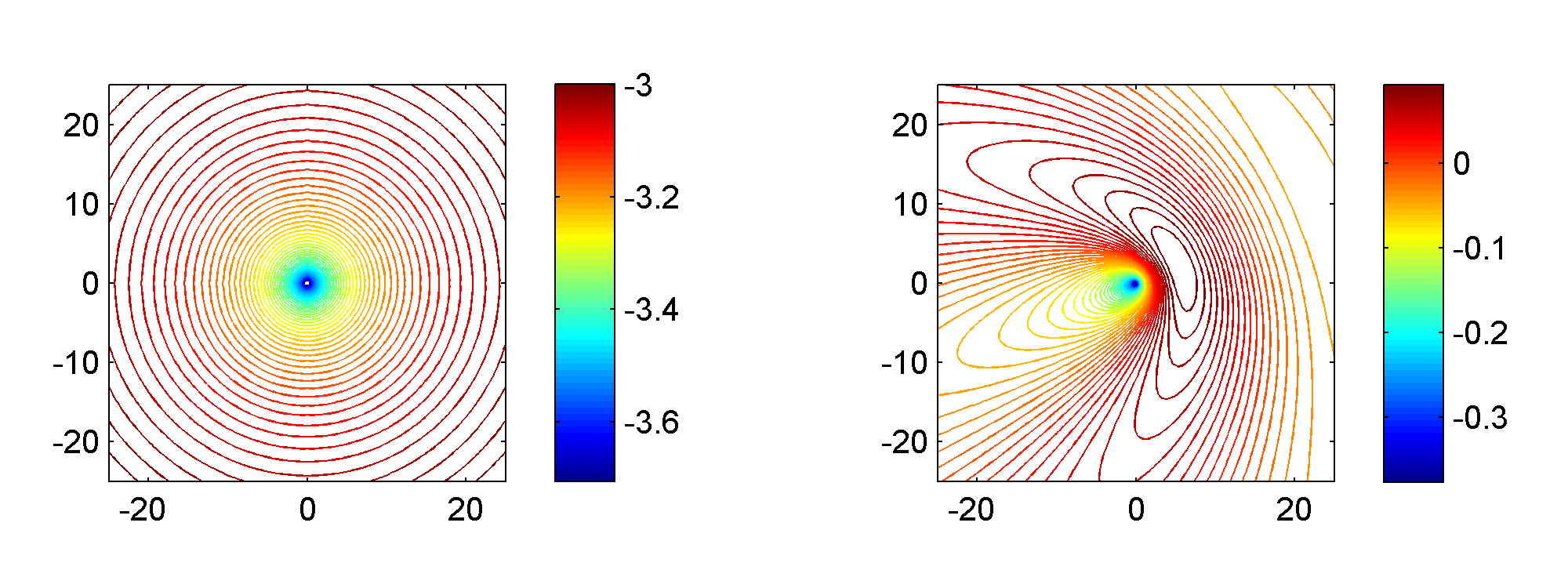}
\caption{Stream function contours for $\psi_G$ in the $\beta$\hyp{}plane, for the
 Green's function (\ref{eq:5.11})-(\ref{eq:5.13}) for Rossby waves
in a rotating wind. The parameter $\bar{\beta}=\beta L/\Omega=0.5$,
describes both the rotation rate of the wind ($\Omega$) and the Coriolis
effect by the Rossby shear parameter $\beta$. The Rossby deformation radius
$R_d\to\infty$ ($\bar{k}_d=0$). The left panel is at time $t=0$
and right panel is at time $\Omega t=\pi/4$.}\label{fig:rossby6}
\end{center}
\end{figure}
\begin{figure}[ht]
\begin{center}
\includegraphics[width=16cm]{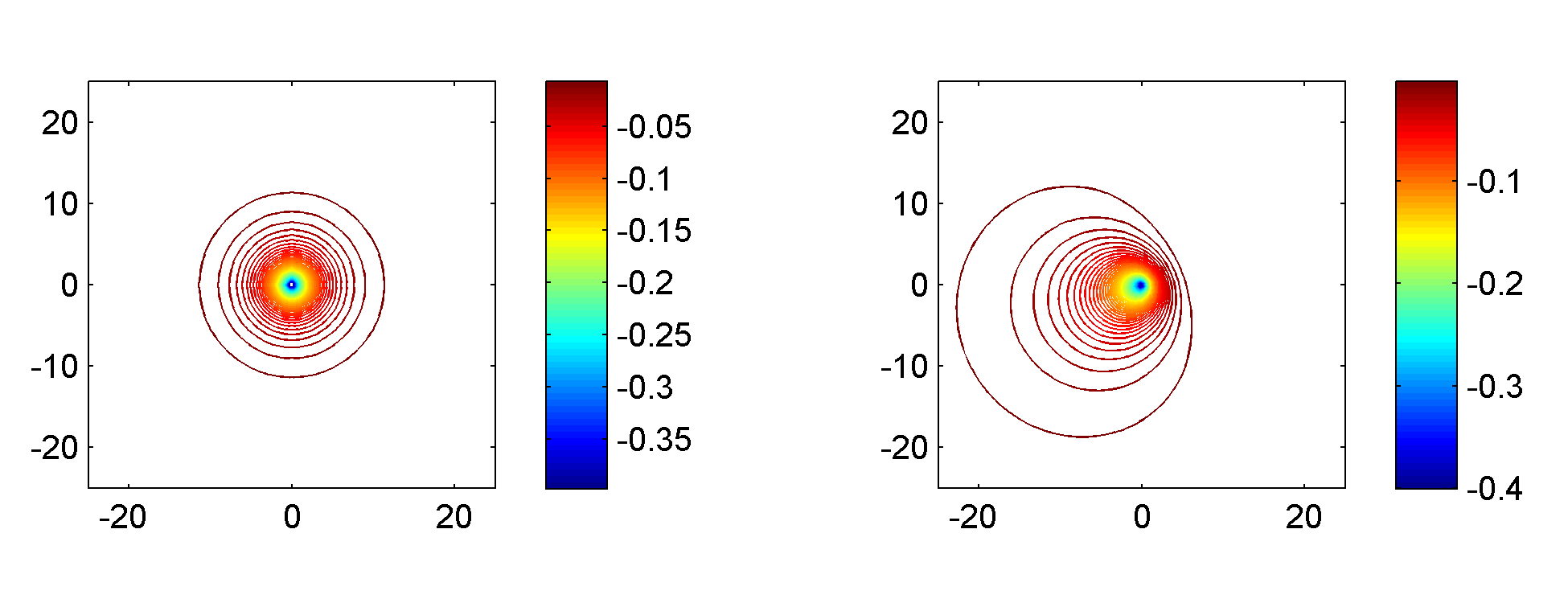}
\caption{Stream function contours for $\psi_G$ in the $\beta$\hyp{}plane, for the
 Green's function (\ref{eq:5.11})-(\ref{eq:5.13}) for Rossby waves
in a rotating wind. The parameter $\bar{\beta}=\beta L/\Omega=0.5$,
describes both the rotation rate of the wind ($\Omega$) and the Coriolis
effect by the Rossby shear parameter $\beta$. The Rossby deformation radius
$R_d=1/k_d$ where ${\bar k}_d=0.25$ and $k_d=\bar{k_d}/L$.  The left panel is at time $t=0$
and right panel is at time $\Omega t=\pi/4$.}\label{fig:rossby7}
\end{center}
\end{figure}

The behavior of the Green's function should change according to whether (\romannumeral1)\ $r>\beta t/k_d^2$ or 
(\romannumeral2)\ $r<\beta t/k_d^2$. In the limit as $k_d\to 0$ 
(infinite Rossby deformation radius $R_d\to\infty$), the outer zone (\romannumeral1) does not exist. 
Thus we expect that a finite $R_d$ will be important at large distances $r$ from the source, 
whereas points very close  to the source will not be significantly affected by a finite $R_d$. 

Figure 1 shows a contour plot of the Veronis (1958) Green's function (\ref{eq:v11}) for the stream function 
$\psi$ in the normalized $(x,y)$\hyp{}plane (i.e. the $\beta$\hyp{}plane), at time $\beta t=1$. The $(x,y)$ coordinates can be written in the form:
\begin{equation}
x=\frac{\alpha\cos\phi}{\beta t}, \quad y=\frac{\alpha\sin\phi}{\beta t},\quad\alpha=\beta rt. \label{eq:5.9}
\end{equation}
We set $N=1$. If we choose the parameters:
\begin{equation}
\beta=10^{-11} \rm{m}^{-1} \rm{s}^{-1},\quad  L=(\beta T)^{-1}=1000\rm{km} \label{eq:5.10}
\end{equation}
then the plot corresponds to a time  $t\sim 10^5$ seconds, which is about a day
 ($1 \rm{day}=8.64\times 10^4$ seconds) and a distance $\Delta x=1$ corresponds to $1000\rm{km}$. 
 The contour plot of $\psi$ is similar to Figure 1 of Veronis (1958). 


Figure 2 shows a contour plot of $\psi_G$ in the $(x,y)$~-plane for the case  
 $\bar{k}_d=k_d L=0.25$, which for $L=1000\rm{km}$ basic scale length, corresponds to a Rossby deformation 
radius of  
$R_d=4000 {\rm km}$. In this case, the dividing boundary $r_b(t)=\beta t/k_d^2$ between the inner and 
outer region solutions is at $r_b=16,000{\rm km}$ for the parameters (\ref{eq:5.9}) with 
$\beta t=10^{-6} {\rm m}^{-1}$. 

 Figure 3 shows the time variation of $\psi_G$ for the Green's function (\ref{eq:nowind1})-(\ref{eq:nowind2}) 
for the infinite Rossby radius case ($k_d=0$) at $\bar{r}=r/L=5$, in the north, west and east diections. 
The variation in the south direction is the same as in the north direction, due to the north\hyp{}south
symmetry of the solution. The western solution profile is a monotonic increasing function of $\beta t$. 
The eastern  ($\phi=0$) and northern ($\phi=90^\circ$) solutions oscillate, with decreasing amplitude 
for large $t$, and are tangent to the western solution at specific times ($\beta t\approx 2$ and 
$\beta t\approx 8$ for the eastern solution and $\beta t=4$ for the northern solution). 
From (\ref{eq:v11}), the Veronis Green's function depends on $\alpha=\beta r t$ and $\gamma=\cos(\phi/2)$. 
Thus the $\beta t$ axis could be replaced by the $\beta t r$ axis divided by the constant 
$r$ value (i.e. $\bar{r}=5$). In other words, the plot in Figure 3, could also be re\hyp{}interpreted 
as showing the radial dependence of the solution for a fixed $t$. 

Figure 4 shows the time variation of $\psi_G$ for the Green's function (\ref{eq:nowind1})-(\ref{eq:nowind2})
versus $\beta t$ for a fixed $r$ ($\bar{r}=5$) for the case of a 
finite Rossby radius ($\bar{k}_d=0.25$). The figure 
shows the change in the solution compared to the infinite Rossby deformation radius case in Figure 3. 
The amplitude of $\psi_G$ is comparable to the infinite Rossby deformation radius case in figure 
3, but now all of solution curves (for the North, East, West cases) oscillate in time, whereas for the $k_d=0$  
case in Figure 3, the western solution curve is monotonic. 

Figure 5, shows the western, northern and eastern solution curves for the case $k_d=0$ (solid curves) 
as compared to the finite Rossby radius solutions ($\bar{k}_d=0.25$) again for $\bar{r}=5$. Notice 
how the $\bar{k}_d=0.25$ solutions oscillate about the $k_d=0$ solutions. The maxima and minima 
for $\bar{k}_d=0.25$ case, are now displaced from the origin ($\beta t=0$).

\subsection{Rotating wind Green's function}
 In this section, we investigate the rotating wind Green's function (\ref{eq:rotate1})-(\ref{eq:rotate3}). 
The solution, can be written as a dimensionless integral over $\bar{k}=kL$ where $L$ is 
a characteristic scale length, i.e.,
\begin{equation}
\psi_G=-\frac{N}{2\pi} \int_0^\infty d\bar{k} \frac{\bar{k}}{\bar{k}^2+\bar{k}_d^2} J_0(D), \label{eq:5.11}
\end{equation}
where
\begin{align}
D=&\left\{\left[\bar{k}\bar{x}+\bar{a}\sin\bar{t}\right]^2
+\left[\bar{k}\bar{y}+\bar{a}-\bar{a}\cos\bar{t}\right]^2\right\}^{1/2}, \label{eq:5.12}\\
\bar{a}=&\frac{\bar{k}\bar{\beta}}{\bar{k}^2+\bar{k}_d^2}, \quad \bar{\beta}=\frac{\beta L}{\Omega},\quad 
\bar{t}=\Omega t. \label{eq:5.13}
\end{align}
A typical value of the parameter $\bar{\beta}$  corresponds to $\Omega=2\times 10^{-5}\rm{s}^{-1}$, 
$L=1000\rm{km}$, and $\beta=10^{-11}{\rm m}^{-1}\rm{s}^{-1}$. This results in 
$\bar{\beta}=10^{-11}\times 10^6/(2\times 10^{-5})=0.5$ as a typical value for $\bar{\beta}$. Here, 
we take $\Omega=V_w/L$ corresponding to an azimuthal wind in the $\beta$\hyp{}plane in which 
$V_w=80 {\rm km}\ {\rm hr}^{-1}$ and $L=1000 {\rm km}$. 

Figure 6 (left panel) shows a contour plot of the Rossby wave Green's function (\ref{eq:5.11}) 
for the parameter values $\bar{\beta}=0.5$, $N=1$, at time $\bar{t}=0$ for the case $\bar{k}_d=0$ 
(infinite Rossby deformation radius), and the right panel, shows the solution at time $\Omega t=\bar{t}=\pi/4$.
The main point to note is the rotation of the contours at time $\bar{t}=\pi/4$ compared with the contours 
at time $\bar{t}=0$. Figure 7, shows  similar contour plots  of $\psi_G$ 
for the case $\bar{k}_d=0.25$ 
at time $\bar{t}=0$ (left panel), and   
 $\psi_G$ at time $\bar{t}=\Omega t=\pi/4$ (right panel). The case with $k_d=0.25$ is similar 
to the contour plots for $k_d=0$ (Figure 6) except that the contours are more confined and compressed
than the $k_d=0$ contours.

\section*{6.\ Concluding Remarks}
In this paper we obtained Green's function solutions for Rossby waves on a rotating planet, 
using the $\beta$\hyp{}plane approximation. The solutions were obtained for the case 
in which the background flow includes a local rotating wind (e.g. McKenzie and Webb (2015)), 
in which the local background fluid velocity ${\bf U}=U_\phi {\bf e}_\phi=r\Omega {\bf e}_\phi$
in the  $\beta$\hyp{}plane. The local $(x,y)$ coordinates in the $\beta$\hyp{}plane  are Cartesian 
coordinates corresponding to the east ($x$) and north ($y$) directions where $x=r\cos\phi$, $y=r\sin\phi$,
and ${\bf e}_\phi=(-\sin\phi,\cos\phi,0)$ and ${\bf e}_z=(0,0,1)^T$ is the local vertical direction 
at latitude $\theta$, and $\Omega$ is assumed constant.  

We first considered  the Green's function for the case of no local wind ($\Omega=0$), with  a delta 
function source term  $Q=N\delta(x)\delta(y)\delta(t)$ in the Rossby wave vorticity equation (\ref{eq:2.6})
(Section 2). This Green's  function was studied by Veronis (1958) who was interested in 
oceanic Rossby waves. We obtain a different, but equivalent  form for the Green's function. 
 Veronis (1958) only studied in detail the Green's 
function for the case of an infinite Rossby deformation radius 
($k_d=0$ or $R_d\to\infty$) for which the solution could be reduced 
to a tractible form in terms of an integral involving a Bessel function of the first kind 
of order zero. Our analysis gives the Green's function, in a tractible form, for both 
$k_d\neq 0$ and for $k_d=0$. We show the equivalence of our Green's function 
to the Veronis form of the Green's function for $k_d=0$ (Appendices B and C), 
 after a typographical error in the Veronis solution is accounted for. 
The difference in the two forms of the solutions occurs because of the different integration
variables used in the two formulations. 

A similar Fourier transform method was used to determine the Green's function $\psi$ for the 
stream function in Section 4, for 
the case of a local rotating wind in the $\beta$\hyp{}plane, with constant angular velocity 
$\Omega$ about the local vertical direction (the $z$\hyp{}axis), 
in which the wind rotates from east to north. The rotation operator about the local $z$ axis, under Fourier 
transformation, transforms to the rotation operator $-\Omega(\partial/\partial\Phi)$ in ${\bf k}$ 
space, where ${\bf k}=(k_x,k_y)=k\left(\cos\Phi,\sin\Phi\right)$ is the wave number, and $\Phi$ 
is the azimuthal angle in ${\bf k}$\hyp{}space. The problem reduces to solving a first order ordinary 
differential equation for the transform $\bar{\psi}$, where the independent 
variable is the azimuthal angle $\Phi$  in ${\bf k}$\hyp{}space. The required solution in ${\bf k}$\hyp{}space
is obtained by requiring that the solution is $2\pi$\hyp{}periodic in $\Phi$. The inversion of the transform 
$\bar{\psi}$  is carried out using the generating function identity for Bessel functions of the 
first kind, listed in Abramowitz et al, (1965) and by using a Neumann series identity for the sum 
of products of Bessel functions (Gradshteyn and Ryzhik, 2000). The solution is given as an integral over $k$ and involves 
a zero order Bessel function of argument $D$ (i.e. $J_0(D)$), in which $D$ depends on 
$x$, $y$, $t$ and $\Omega$. In the limit as $\Omega\to 0$ the solution reduces to the Green's function 
solution of section 3, in which $\Omega=0$. 

Illustrative solutions of the Rossby wave Green's function for the case of no rotating wind 
($\Omega=0$) and for the rotating wind case ($\Omega\neq 0$) were given in section 5. The stream function 
$\psi_G$ for an infinite Rossby deformation radius ($k_d=0$) obtained by Veronis (1958) was compared 
with the Rossby wave solution with $R_d=4000km$ ($\bar{k_d}=0.25$) (Figures 1 and 2). 
The direction of the fluid velocity is parallel to the contours of the stream function in these 
examples (${\bf u}=(-\psi_y,\psi_x)$ in the $\beta$ plane). The Green's function solutions for 
the stream function $\psi_G$ versus the time parameter $\beta t$ for both the $k_d=0$ 
and $\bar{k}_d=0.25$ and $\bar{r}=5$,  were displayed in Figures 3-5. In the infinite Rossby wave 
deformation radius limit ($k_d=0$), $\psi_G$ has the functional form $\psi_G=\psi_G(\beta r t,\cos(\phi/2))$
(see Figure 3). The solution for $\psi_G$ for $k_d\neq0$ (finite $R_d$) has the functional form:
$\psi_G=\psi_G(\beta r t,k_d r,\cos(\phi/2))$. Figure 5 shows that the solution for $\bar{k}_d=0.25$
oscillates about the $k_d=0$ solution. The solution is symmetric about the east\hyp{}west axis (i.e. the 
northern solution is the same as the southern solution). 

Figure 6 shows contour plots of $\psi_G$ for the rotating wind Green's function  
(\ref{eq:rotate1})-(\ref{eq:rotate3}) at time $t=0$ and at time $\Omega t=\pi/4$.
The initial Green's function at time $t=0$ is circularly symmetric in the $\beta$\hyp{}plane, 
but is rotated and skewed at time $\Omega t=\pi/4$. The contours are symmetric about the 
oblique symmetry axis, are flattened to the East, and 
elongated to the west. The corresponding rotating wind 
Green's function for a finite Rossby deformation radius (${\bar k}_d=0.25$) at times $t=0$ and 
$\Omega t=\pi/4$ are shown in Figure 7. These figures give a brief overview of the rotating wind 
Green's function, but further details of the solution for $\psi_G$ remain to be explored.

\section*{Acknowledgements}
We acknowledge stimulating discussions with the late J.F. McKenzie, who first suggested  
investigating the rotating wind Rossby wave problem. 
 GMW was supported in part by NASA grant NNX15A165G.

\appendix
\section*{Appendix A}
\setcounter{section}{1}
\setcounter{equation}{0}

In this appendix we derive the solution (\ref{eq:3.25}), $S=A$, of the wave eikonal equation (\ref{eq:3.23}) 
by the method of characteristics (e.g. Sneddon (1957), Courant and Hilbert (1989), Vol. 2).
We first write (\ref{eq:3.23}) in the form:
\begin{equation}
F=r(p^2+q^2+k_d^2)-\beta p=0, \label{eq:A1}
\end{equation}
where
\begin{equation}
r=S_t,\quad p=S_x,\quad q=S_y. \label{eq:A2}
\end{equation}
The characteristics of nonlinear, first order partial differential equations of the 
form (\ref{eq:A1}) (Sneddon (1957), Courant and Hilbert (1989)) are given by:
\begin{align}
\frac{dt}{d\tau}=&F_r,\quad \frac{dx}{d\tau}=F_p,\quad \frac{dy}{d\tau}= F_q, 
\quad \frac{dS}{d\tau}=r F_r+p F_p+ q F_q,\nonumber\\
\frac{dr}{d\tau}=&-\left(F_t+r F_S\right),\quad \frac{dp}{d\tau}=-(F_x+p F_S), \quad 
\frac{dq}{d\tau}=-(F_y+q F_S), \label{eq:A3}
\end{align}
where $F_\phi=\partial F/\partial\phi$ is the partial derivative of $F$ with respect to $\phi$, 
and $\tau$ is a parameter along the characteristics. 

Using (\ref{eq:A1}) to evaluate the derivatives in (\ref{eq:A3}), we obtain:
\begin{equation}
\frac{dr}{d\tau}=0,\quad \frac{dp}{d\tau}=0,\quad \frac{dq}{d\tau}=0, \label{eq:A4}
\end{equation}
with solutions:
\begin{equation}
r=r_0,\quad p=p_0,\quad q=q_0, \label{eq:A5}
\end{equation}
where $r_0$, $p_0$ and $q_0$ are constants on the characteristics, subject to the 
partial differential equation constraint:
\begin{equation}
r_0(p_0^2+q_0^2+k_d^2)-\beta p_0=0, \label{eq:A6}
\end{equation}
i.e. $(r_0,p_0,q_0)$ must satisfy the original partial differential equation (\ref{eq:A1}). 

Computing the remaining derivatives in (\ref{eq:A3}) gives:
\begin{align}
\frac{dx}{d\tau}=&2rp-\beta,\quad \frac{dy}{d\tau}=2rq,\quad \frac{dt}{d\tau}=p^2+q^2+k_d^2, \nonumber\\
\frac{dS}{d\tau}=&p(2pr-\beta)+q(2rq)+r(p^2+q^2+k_d^2)=2r(p^2+q^2)+F\equiv 2r(p^2+q^2), \label{eq:A7}
\end{align}
where we have used the fact that $F=0$ on the characteristics. Because $r$, $p$ and $q$ are constants
on the characteristics, the integration of (\ref{eq:A7}) is straightforward. We obtain the integrals:
\begin{align}
t=&(p_0^2+q_0^2+k_d^2)\tau,\quad x=(2 r_0 p_0-\beta)\tau,\nonumber\\
y=&2 r_0 q_0\tau,\quad S=2r_0(p_0^2+q_0^2)\tau, \label{eq:A8}
\end{align}
where we have chosen the integration constants such that $t$, $x$, $y$ and $S$ are all zero at $\tau=0$. 
From (\ref{eq:A8}) we obtain the equations:
\begin{align}
x+\beta\tau=&2 r_0 p_0\tau,\quad y=2r_0 q_0\tau,\nonumber\\
S^2=&4 r_0^2\left(p_0^2+q_0^2\right)^2\tau^2,\quad p_0^2+q_0^2=k^2. \label{eq:A9}
\end{align}
The last equation in (\ref{eq:A9}) states that $S_x^2+S_y^2=k^2$. Also note from (\ref{eq:A9}) that:
\begin{equation}
k^2\left[(x+\beta\tau)^2+y^2\right]=\left(p_0^2+q_0^2\right)\left[4 r_0^2(p_0^2+q_0^2)\tau^2\right]
= S^2. \label{eq:A10}
\end{equation}
Finally using (\ref{eq:A8}), (\ref{eq:A10}) reads:
\begin{equation}
S^2=(k x+\tilde{\beta} t)^2+k^2 y^2, \label{eq:A11}
\end{equation}
which is the wave eikonal solution (\ref{eq:3.25}), which shows that (\ref{eq:3.25}) can be obtained 
by integrating the Cauchy characteristics.

\section*{Appendix B}
\setcounter{section}{2}
\setcounter{equation}{0}
In this appendix we discuss other forms of the Veronis (1958) Rossby wave 
Green's function, obtained for the case of an infinite Rossby deformation radius
($k_d\to 0$, $R_d=1/k_d\to\infty$). In the limit as 
$k_d\to 0$, the Green's function $\psi_G$ in (\ref{eq:v10})
reduces to:
\begin{equation}
\psi_G=-\frac{N}{2\pi} \int_{c-i\infty}^{c+i\infty} \frac{ds}{2\pi i}
\exp\left(st-\frac{\beta x}{2s}\right) \frac{1}{s} K_0\left(\frac{\beta r}{2s}\right), \label{eq:B1}
\end{equation}
where $N$ is the normalization constant 
for the Green's function with source $Q=N\delta (x)\delta (y)\delta (t)$.
We split the transform:
\begin{equation}
\bar{\psi}(s)=-\frac{N}{2\pi} \exp\left(-\frac{\beta x}{2s}\right)\frac{1}{s} 
K_0\left(\frac{\beta r}{2s}\right), 
\label{eq:B2} 
\end{equation}
into two separate functions as:
\begin{align}
\bar{\psi}(s)=&-\frac{N}{2\pi} F_1(s)F_2(s),\nonumber\\ 
F_1(s)=&\frac{1}{\sqrt{s}}\exp\left[-\frac{\beta(x+r)}{2s}\right],\nonumber\\
F_2(s)=&\frac{1}{\sqrt{s}}\exp\left(\frac{\beta r}{2s}\right) 
K_0\left(\frac{\beta r}{2s}\right). \label{eq:B3}
\end{align}
The inverse Laplace transform of (\ref{eq:B1}) may be written as a convolution integral:
\begin{align}
\psi_G=&-\frac{N}{2\pi} \int_0^t f_2(\tau)f_1(t-\tau)\ d\tau\label{eq:B4}\\
\equiv&-\frac{N}{2\pi} \int_0^t \left\{\frac{2}{\sqrt{\pi\tau}} K_0\left(2\sqrt{\beta r\tau}\right)\right\}
\left\{\frac{1}{\sqrt{\pi(t-\tau)}}\cos\left(\sqrt{2\beta (x+r)(t-\tau)}\right)\right\} d\tau. \label{eq:B5}
\end{align}
 (\ref{eq:B5}) for $\psi_G$ is the Green's function formula (20) of Veronis (1958). 
The expressions for $f_2(\tau)$ and $f_1(t-\tau)$ in (\ref{eq:B5}) follow from the Laplace transforms:
\begin{align}
{\cal L}^{-1}\left[\frac{1}{\sqrt{p}}\exp\left(-\frac{\alpha}{p}\right)\right]
=&\frac{1}{\sqrt{\pi t}}
\cos\left(2\sqrt{\alpha t}\right), \nonumber\\
{\cal L}^{-1}\biggl[\frac{1}{\sqrt{p}}
\exp\left(\frac{\alpha}{p}\right)
K_0\left(\frac{\alpha}{p}\right)\biggr]
=&\frac{2}{\sqrt{\pi t}} K_0\left(\sqrt{8\alpha t}\right), \label{eq:B6}
\end{align}
where ${\cal L}^{-1}$ denotes the inverse Laplace transform operation [e.g. Erdelyi et al. (1954), 
Vol. 1, formula (37), p.345; and formula (33), p282]. 

By using the transformations:
\begin{equation}
\tau=t\sin^2\Theta,\quad x=r\cos\phi, \label{eq:B8}
\end{equation}
and by introducing the variables:
\begin{equation}
\alpha=\beta r t,\quad \gamma=\cos(\phi/2), \label{eq:B9}
\end{equation}
(note $x+r=2r\gamma^2$), the convolution integral (\ref{eq:B5}) reduces to:
\begin{align}
\psi_G=&-\frac{2N}{\pi^2}\int_0^{\pi/2} K_0(\sqrt{4\alpha}\sin\Theta)\cos(2\sqrt{\alpha}\gamma\cos\Theta)\ d\Theta
\nonumber\\ 
=&-\frac{N}{\pi^2}\int_0^{\pi} K_0(\sqrt{4\alpha}\sin\Theta)\cos(2\sqrt{\alpha}\gamma\cos\Theta)\ d\Theta.
\label{eq:B10}
\end{align}
Using Basset's integral:
\begin{equation}
K_0(r\delta)=\int_0^\infty \frac{\cos(ru)}{\sqrt{u^2+\delta^2}} du, \label{eq:B11}
\end{equation}
(Abramowitz and Stegun (1965), formula 9.6.21, p. 376, with $x\to r\delta$, $t\to u/\delta$, see also 
Veronis (1958)], we obtain:
\begin{equation}
K_0(\sqrt{4\alpha}\sin\Theta)=\int_0^\infty \frac{\cos(u\sin\Theta)}{\sqrt{u^2+4\alpha}} du, \label{eq:B12}
\end{equation}
as an integral representation of the $K_0$ Bessel term in (\ref{eq:B10}). 

Using (\ref{eq:B12}) in (\ref{eq:B10}) and switching the order of integration gives:
\begin{equation}
\psi_G=-\frac{N}{\pi^2}\int_0^\infty du\ \frac{g(u,\alpha,\gamma)}{\sqrt{u^2+4\alpha}}, \label{eq:B13}
\end{equation}
and 
\begin{align}
g(u,\alpha,\gamma)=&\int_0^{\pi} \cos(u\sin\Theta)\cos(2\sqrt{\alpha}\gamma\cos\Theta)\ d\Theta\nonumber\\
\equiv&\frac{1}{2}\int_0^{\pi}\biggl\{\cos(u\sin\Theta-2\sqrt{\alpha}\gamma\cos\Theta)
+\cos(u\sin\Theta+2\sqrt{\alpha}\gamma\cos\Theta)\biggr\}\ d\Theta\nonumber\nonumber\\
=&\frac{1}{2}\int_0^{\pi}\left\{\cos[A\sin(\Theta-\epsilon)]+\cos[A\sin(\Theta+\epsilon)]\right\}\ d\Theta.
\label{eq:B14}
\end{align}
where
\begin{equation}
A\cos\epsilon=u,\quad A\sin\epsilon=2\sqrt{\alpha}\gamma, \quad
A=\left(u^2+4\gamma^2\alpha\right)^{1/2}. 
\label{eq:B15}
\end{equation}
Using the Bessel function generating identity in the form:
\begin{equation}
\cos\left(z\sin\theta\right)=\sum_{n=-\infty}^\infty J_n(z) \cos(n\theta), \label{eq:B16}
\end{equation}
with $z\to A$ and $\theta\to\Theta\pm \epsilon$, (\ref{eq:B14}) reduces to:
\begin{align}
g(u,\alpha,\gamma)=&\frac{1}{2}\int_0^{\pi} \sum_{n=-\infty}^{\infty} J_n(A)
\left\{\cos[n(\Theta+\epsilon)]+\cos[n(\Theta-\epsilon)]\right\} d\Theta\nonumber\\
=&\sum_{n\neq 0} J_n (A)\cos(n\epsilon)\frac{\sin(n\pi)}{n}+\pi J_0(A)
\equiv \pi J_0(A), 
\label{eq:B17}
\end{align}
(note that $\sin(n\pi)=0$).

Using (\ref{eq:B17}) in (\ref{eq:B13}) gives:
\begin{align}
\psi_G=&-\frac{N}{\pi}\int_0^\infty \frac{du}{\sqrt{u^2+4\alpha}} 
J_0\left[\sqrt{u^2+4\gamma^2\alpha}\right]\nonumber\\
\equiv&-\frac{N}{\pi}\int_0^\infty \frac{dz}{\sqrt{1+z^2}} 
J_0\left[2\sqrt{\alpha}\left(z^2+\gamma^2\right)^{1/2}\right],  \label{eq:B20}
\end{align}
where we used the transformation of variables $z=u/\sqrt{4\alpha}$. The solution (\ref{eq:B20}) 
is essentially (21) of Veronis (1958), except $\sqrt{\alpha}$ is replaced by $\alpha$ 
in his equation (21). This means that $\alpha$ must be replaced by $\sqrt{\alpha}$ in Veronis's 
equation (22).

\section*{Appendix C}
\setcounter{section}{3}
\setcounter{equation}{0}
In this appendix we show that the Green's function $\psi_V$ obtained by Veronis for the case $k_d=0$,
given in (\ref{eq:v11}) is equivalent to the Green's function in proposition \ref{prop31} 
and in (\ref{eq:nowind1}) and (\ref{eq:nowind2}). From (\ref{eq:v11}), 
\begin{equation}
\psi_V=-\frac{N}{\pi}\int_0^\infty \frac{dz}{\sqrt{1+z^2}} 
J_0\left[2\sqrt{\alpha}\left(z^2+\gamma^2\right)^{1/2}\right]. \label{eq:C1}
\end{equation}
This is equivalent to the Green's function (\ref{eq:nowind1})-(\ref{eq:nowind2}) with $k_d=0$:
\begin{equation}
\psi_G=-\frac{N}{2\pi}\int_0^\infty \frac{dk}{k} J_0(A), \label{eq:C2}
\end{equation}
where
\begin{equation}
A=\left[\left(k x+\frac{\beta t}{k}\right)^2+k^2 y^2\right]^{1/2}. \label{eq:C3}
\end{equation}
In the Veronis (1958) paper, $\sqrt{\alpha}\to \alpha$ in (\ref{eq:C1}). This is a typographical 
error in Veronis' equation (21).
One expects that there is a transformation of the integration variable, that maps the Green's function 
(\ref{eq:C1}) onto the Green's function (\ref{eq:C2})-(\ref{eq:C3}), in which the argument of the 
$J_0$ Bessel function does not change. This will be the case, if:
\begin{equation}
\left[4\alpha(z^2+\gamma^2)\right]^{1/2}
=\left[\left(kx+\frac{\beta t}{k}\right)^2+k^2 y^2\right]^{1/2}\equiv A. 
\label{eq:C4}
\end{equation}
Squaring (\ref{eq:C4}) gives:
\begin{equation}
4\alpha\left(z^2+\gamma^2\right)
=k^2 r^2 +\left(\frac{\beta t}{k}\right)^2+2\beta t x. \label{eq:C5}
\end{equation}
Noting that
\begin{equation}
4\alpha\gamma^2\equiv 4\beta r t\gamma^2=4\beta r t\cos^2(\theta/2)= 2\beta t(x+r), \label{eq:C6}
\end{equation}
and  using (\ref{eq:C6}) and (\ref{eq:C5}) gives the equation:
\begin{equation}
\left(kr-\frac{\alpha}{kr}\right)^2=4\alpha z^2. \label{eq:C7}
\end{equation}
Taking the square root of (\ref{eq:C7}) gives:
\begin{equation}
\sqrt{4\alpha} z=\sigma \left(kr-\frac{\alpha}{kr}\right)\quad\hbox{where}\quad \sigma=\pm1 . 
\label{eq:C8}
\end{equation}
An inspection of (\ref{eq:C8}) reveals that $z=0$ when $k=k_c$ where
\begin{equation}
k_c r=\sqrt{\alpha}. \label{eq:C9}
\end{equation}
By sketching the graph of $z$ versus $kr$ in (\ref{eq:C8}) shows that it is necessary to choose the 
$\sigma=1$ branch for $k>k_c$ and to choose the $\sigma=-1$ branch for $0<k<k_c$. These choices ensure
that $z$ is positive (i.e. $0<z<\infty$). Thus, transformations (\ref{eq:C8}) can be written as:
\begin{equation}
\sqrt{4\alpha}z=\biggl\{\begin{array}{cccc}
\alpha/(kr)-kr&\hbox{if}& k<k_c&(\sigma=-1),\\
kr-\alpha/(kr)&\hbox{if}&k>k_c&(\sigma=1).
\end{array}\biggr. \label{eq:C9a}
\end{equation}

Equation (\ref{eq:C8}) can be expressed as a quadratic equation for $kr$ as:
\begin{equation}
(kr)^2-\sigma \sqrt{4\alpha} z (kr)-\alpha=0, \label{eq:C10}
\end{equation}
with solutions:
\begin{equation}
kr=\sqrt{\alpha}\left(\sigma z\pm \sqrt{z^2+1}\right). \label{eq:C11}
\end{equation}
Since we require $kr>0$ and $z>0$ in (\ref{eq:C1}) and (\ref{eq:C2}) we choose the transformations:
\begin{equation}
kr=\biggl\{\begin{array}{ccc}
\sqrt{\alpha}(\sqrt{z^2+1}-z)&\hbox{for}&k<k_c,\\
\sqrt{\alpha}(z+\sqrt{z^2+1})&\hbox{for}&k>k_c. 
\end{array}\biggr. 
\label{eq:C12}
\end{equation}
From the transformations (\ref{eq:C9a}) we obtain:
\begin{equation}
dz=\frac{\sigma}{\sqrt{4\alpha}}\frac{dk}{k}\left(kr+\frac{\alpha}{kr}\right),
\quad \sqrt{z^2+1}=\frac{1}{\sqrt{4\alpha}}\left(kr+\frac{\alpha}{kr}\right), 
\quad \frac{dz}{\sqrt{z^2+1}}=\sigma
\frac{dk}{k}. \label{eq:C14}
\end{equation}

The Green's function $\psi_G$ in (\ref{eq:C2}) can be written in the form:
\begin{align}
\psi_G=&-\frac{N}{2\pi}\left(\int_0^{k_c}+\int_{k_c}^\infty\right)\frac{dk}{k} J_0(A)\nonumber\\
\equiv&-\frac{N}{2\pi}\left[I(0,k_c)+I(k_c,\infty)\right], \label{eq:C15}
\end{align}
where $I(0,k_c)$ is the integral from $k=0$ to $k=k_c$ and $I(k_c,\infty)$ is the integral over 
the range $k_c<k<\infty$. Using (\ref{eq:C4}) and (\ref{eq:C14}) we obtain:
\begin{align}
I(0,k_c)=&\int_0^\infty \frac{dz}{\sqrt{z^2+1}} J_0\left(\sqrt{4\alpha(z^2+\gamma^2)}\right),\nonumber\\
I(k_c,\infty)=&\int_{0}^\infty \frac{dz}{\sqrt{z^2+1}} J_0\left(\sqrt{4\alpha(z^2+\gamma^2)}\right). 
\label{eq:C16}
\end{align}
Substituting (\ref{eq:C16}) into (\ref{eq:C15}) we obtain:
\begin{equation}
\psi_G=\psi_V, \label{eq:C17}
\end{equation}
which proves that the Green's function $\psi_G$ of (\ref{eq:nowind1}) and (\ref{eq:nowind2}) 
is equivalent to the Green's function $\psi_V$ in (\ref{eq:C1}) obtained by Veronis (1958). 
It is clear from the above analysis, that 
\begin{equation}
\psi_G=-\frac{N}{\pi}\int_0^{k_c} \frac{dk}{k} J_0(A), \label{eq:C18}
\end{equation}
is an alternative formula for $\psi_G$.

\end{document}